\documentclass[11pt]{article}
\usepackage{tikz}
\usetikzlibrary{fit,positioning}
\usetikzlibrary{arrows}
\usetikzlibrary{automata}
\usetikzlibrary{calc}
\usetikzlibrary{decorations.pathmorphing}
\usetikzlibrary{decorations.pathreplacing}
\tikzset{every state/.style={minimum size=0pt}}
\tikzset{
    position/.style args={#1:#2 from #3}{
	at=(#3.#1), anchor=#1+180, shift=(#1:#2)
    }
}

\usepackage{complexity}
\usepackage{bm}
 \usepackage{amsmath,mathtools}
\usepackage{amssymb}
\usepackage[top=1in, bottom=1.25in, left=1.25in, right=1.25in]{geometry}

\usepackage{amsthm}
\usepackage{cleveref}
\theoremstyle{plain}
\newtheorem{theorem}{Theorem}[section]
\newtheorem{lemma}[theorem]{Lemma}
\newtheorem{proposition}[theorem]{Proposition}
\theoremstyle{definition}
\newtheorem{definition}[theorem]{Definition}
\theoremstyle{remark}
\newtheorem{remark}[theorem]{Remark}
\newtheorem{example}[theorem]{Example}

\newcommand{\eqdef}{\stackrel{\mbox{\begin{scriptsize}def\end{scriptsize}}}{=}}
\newcommand{\equivdef}{\mathrel{\stackrel{\mbox{\begin{scriptsize}def\end{scriptsize}}}{\Longleftrightarrow}}}
\newcommand{\obracew}[2]{{\overset{#2}{\overbrace{#1}}}}
\newcommand{\tuple}[1]{\langle #1 \rangle}
\newcommand{\step}[1]{\:{\xrightarrow{\!#1\!}}\:}
\newcommand{\relstep}[1]{\:{\xrightarrow{\!#1\!}_{\text{rel}}}\:}
\newcommand{\subword}{\preccurlyeq}               
          
\newcommand{\supword}{\succcurlyeq}               
          
\newcommand{\alphabet}{\mathop{\mathrm{alph}}\nolimits}
\newcommand{\post}{\mathop{\mathtt{Post}}\nolimits}
\newcommand{\Pre}{\mathop{\mathtt{Pre}}\nolimits}
\newcommand{\pr}{\mathop{\mathtt{pr}}\nolimits}
\newcommand{\rea}{\mathop{\mathtt{rea}}\nolimits}
\newcommand{\size}[1]{\mathopen{\|}#1\mathclose{\|}}
\newcommand{\upc}{\mathop{\uparrow}\nolimits}
\newcommand{\wri}{\mathop{\mathtt{wri}}\nolimits}
\newcommand{\Act}{\mathit{Act}}
\newcommand{\Conf}{\mathit{Conf}}
\renewcommand{\epsilon}{\varepsilon}
\renewcommand{\phi}{\varphi}
\renewcommand{\emptyset}{\varnothing}

\newcommand{\Nat}{\mathbb{N}}
\newcommand{\Rat}{\mathbb{Q}}
\newcommand{\bmkappa}{{\bm{\kappa}}}
\newcommand{\GaD}{\!\!{\implies}\!\!}
\newcommand{\DaG}{\!\!{\impliedby}\!\!}
\newcommand{\ttdol}{\texttt{\$}}
\newcommand{\ttf}{\texttt{f}}
\newcommand{\ttv}{\texttt{v}}
\newcommand{\ttx}{\texttt{x}}
 
\title{On flat lossy channel machines}
\author{Ph.\ Schnoebelen}
\date{}

\begin{document}

\maketitle

\begin{abstract}
We show that reachability, repeated reachability, nontermination and
unboundedness are $\NP$-complete for Lossy Channel Machines that are
\emph{flat}, i.e., with no nested cycles in the control graph.
The upper complexity bound relies on a fine analysis of iterations of lossy
channel actions and uses compressed word techniques for efficiently
reasoning with paths of exponential lengths.
The lower bounds already apply to acyclic or single-path machines.
\end{abstract}

\section{Introduction}
\label{sec-intro}

\noindent
\emph{Lossy channel machines,}
aka LCMs, are FIFO automata, i.e., finite-state machines operating on
buffers with FIFO read/write discipline, where the buffers are
\emph{unreliable}, or \emph{lossy}, in the sense that letters (or
``messages'') in a buffer can be lost nondeterministically at any
time.

LCMs were first introduced as a model for communication protocols
designed to work properly in unreliable environments. They immediately
attracted interest because, unlike FIFO automata with reliable
buffers, they have decidable safety and termination
problems~\cite{finkel94,abdulla96b,cece95,abdulla-forward-lcs}.  It
was later found that LCMs are a relevant computational model
\textit{per se}, useful for verifying timed
automata~\cite{abdulla-icalp05,lasota2008}, modal
logics~\cite{gabelaia06}, etc., and connected to other problems in
computer science~\cite{KS-msttocs,CS-fossacs10,schmitz-toct2016}.

\medskip

\noindent
\emph{Flat LCMs.}
In this paper we consider the case of \emph{flat LCMs}, i.e., LCMs
where the control graph has no nested cycles. In the area of
infinite-state systems verification, flat systems were first
considered in~\cite{fribourg97b,comon98b} for counter
systems\footnote{
   Flatness remains relevant with \emph{finite-state systems}, see
   e.g.,~\cite{kuhtz2011}. This is especially true when
   one is considering the verification of properties expressed in a
   rich logic as in, e.g.,~\cite{decker2017}. In language theory, flat
   finite-state automata correspond to regular languages of polynomial
   density, sometimes called \emph{sparse languages}, or also
   \emph{bounded languages}.
}. In addition, some earlier ``loop acceleration'' results,
e.g.~\cite{boigelot94}, where one can compute reachability sets along
a cycle, can often be generalised to flat systems.
Positive results on flat counter systems can be found
in~\cite{leroux2005b,bozga2009b,demri2010,bozga2014,leroux2014b,demri2015},
and in~\cite{ganty2015} for counter systems with recursive calls.
Regarding flat FIFO automata, verification was shown decidable by
Bouajjani and Habermehl~\cite{bouajjani99b} who improved on earlier
results by Boigelot~\cite{boigelot99b}, and the main verification
problems were only recently proven to be
$\NP$-complete~\cite{esparza2012,finkel2019}.
These results have applications beyond flat systems in the context of
\emph{bounded verification} techniques, where one analyses a bounded
subset of the runs of a general system~\cite{esparza2012}. \\

Flat LCMs have not been explicitly considered in the literature. They
are implicit in forward analysis methods based on loop acceleration,
starting with~\cite{abdulla-forward-lcs}, but these works do not
address the overall complexity of the verification problem, only the
complexity of elementary operations.

It is not clear whether one should expect flat LCMs to be simpler than
flat FIFO automata (on account of unrestricted LCMs being simpler than
the Turing powerful, unrestricted FIFO automata), or if they could be
more complex since message losses introduce some nondeterminism that
does not occur when one follows a fixed cycle in a FIFO
system. Indeed, message losses can be seen as hidden implicit loops
that disrupt the apparent flatness of the LCM.

\medskip

\noindent
\emph{Our contribution.}
We analyse the behaviour of the backward-reachability algorithm on
cycles of lossy channel actions and establish a bilinear upper bound
on its complexity. As a consequence, reachability along runs of the
form $\rho_1\sigma_1^*\rho_2\sigma_2^*\ldots\rho_m\sigma_m^*$ where
the $\rho_i,\sigma_i$ are sequences of channel actions, can be decided
in time $n^{O(m)}$. While shortest reachability witnesses can be
exponentially long when the number $m$ of cycles is not bounded,
techniques based on SLP-compressed words allow handling and checking
these witnesses in polynomial time, leading to an $\NP$ algorithm for
flat LCMs. This easily translates into $\NP$ algorithms for
nontermination, repeated reachability, and unboundedness, and in fact
all four problems are
$\NP$-complete.
Thus the restriction to flat systems really brings some
simplification when compared to the very high complexity ---sometimes
undecidability as is the case for unboundedness--- of verification for
unrestricted LCMs~\cite{CS-lics08,SS-icalp11,schmitz-toct2016}.

\begin{remark}[Lossy channel machines vs.\ lossy channel systems]
In line with most works on loop acceleration and verification of flat
systems, we consider lossy channel ``machines'' instead of the more
usual lossy channel ``systems'', i.e., systems where several
independent concurrent machines communicate via shared channels. This is
because a combination of individually flat machines does not lead to a
``flat'' system. Additionally, finite-state concurrent systems
typically have $\PSPACE$-hard verification problems already when they
have no channels and run synchronously, or when they only synchronise
via bounded channels that can hold at most one message~\cite{DLS-jcss-param}.
\qed
\end{remark}

\noindent
\emph{Outline.}
After some technical preliminaries (\Cref{sec-prelim}), we present our
main technical contribution (\Cref{sec-backward-reachability}): we
analyse the computation of predecessors (of some given configuration)
through a cycle iterated arbitrarily many times.  In particular we
show that the backward-reachability analysis of a single cycle reaches
its fixpoint after a bilinear number of iterations. This leads to an
effective bound on the length of the shortest runs between two
configurations.
In \Cref{sec-NP-via-SLP} we show how the previous analysis can be
turned into a nondeterministic polynomial-time algorithmic via the use
of SLP-compressed words for efficiently computing intermediary channel
contents along a run.
In \Cref{sec-nontermination-and-unboundedness} we show how our main
results also apply to termination, repeated reachability, and boundedness.
Finally \Cref{app-NP-hardness} presents reductions showing how the
problems we considered are $\NP$-hard, even for acyclic LCMs or
single-path LCMs.

\medskip

\noindent
\emph{Related work.}
After we circulated our draft proof, we became aware
that a related $\NP$-membership result will be found
in~\cite{finkel2019jv}. There the authors adapt the powerful technique
from~\cite{esparza2012} and encode front-lossy channel systems into
multi-head pushdown automata, from which an $\NP$-algorithm for
control-state reachability in flat machines ensue. Our approach is
lower level, providing a tight bilinear bound on the number of times a cycle
must be visited in the backward-reachability algorithm. Once these
bounds are established, our NP algorithm only needs to guess the
number of times each cycle is visited.

\section{Preliminaries}
\label{sec-prelim}

We consider words $u,v,w,x,y,z,\ldots$ over a finite alphabet
$\Sigma=\{a,b,\ldots\}$. We write $|u|$ for the length of a word and
$\epsilon$ for the empty word.  The set of letters that occur in $u$
is written $\alphabet(u)$.
For a $n$-letter word $u=a_1\cdots
a_n$ and some index $\ell\in\{0,1,\ldots,n\}$, we write
$u_{\leq\ell} \eqdef a_1\cdots
a_{\ell}$ and $u_{> \ell} \eqdef a_{\ell+1}\cdots a_n$ for the $\ell$-th
prefix and the $\ell$-th suffix of $u$. We write $u_{(\ell)}\eqdef
u_{> \ell}\cdot u_{\leq\ell}$ for the $\ell$-th \emph{cyclic shift} of
$u$.

Exponents are used to denote the concatenation of multiple copies of a
same word, i.e., $u^3$ denotes $u\:u\:u$. A \emph{fractional exponent}
$p\in\Rat$ can be used for $u^p$ if $p\cdot|u|$ is a natural
number. E.g., when $a,b,c$ are letters, $(abc)^{\frac{11}{3}}$, or
equivalently $(abc)^{3+\frac{2}{3}}$, denotes $abc\:abc\:abc\:ab$.

We write $u\subword v$ to denote that $u$ is a (scattered) subword of
$v$, i.e., there exist $2m+1$ words
$u_1,\ldots,u_m,x_0,x_1,\ldots,x_m$ such that $u=u_1\cdots u_m$ and
$v=x_0u_1x_1\ldots u_mx_m$. It is well-known that $\subword$ is a
well-founded partial ordering. For a word $x\in\Sigma^*$, we write
$\upc x\eqdef\{y\in\Sigma^*~|x\subword y\}$ to denote the
\emph{upward-closure} of $x$, i.e., the set of all words that contain
$x$ as a (scattered) subword.

\medskip

\noindent
\emph{LCMs.}
In this paper we consider channel machines with a single communication
channel\footnote{See \Cref{app-multiple-channels} for a generalisation
of our results to multi-channel machines.}.
A \emph{lossy channel machine} (LCM) is a tuple
$S=\tuple{Q,\Sigma,\Delta}$ where $Q=\{p,q,\ldots\}$ is a finite set
of \emph{control locations}, or just ``locations'',
$\Sigma=\{a,b,\ldots\}$ is the finite \emph{message alphabet}, and
$\Delta\subseteq Q\times (\{!,?\}\times\Sigma^*) \times Q$ is a finite
set of  \emph{transition rules}. A rule $\delta=\tuple{q,(d,w),q'}$ has
a start location $q$, an end location $q'$ and a \emph{channel action}
$(d,w)$.  We write $\Act_\Sigma\eqdef \{!,?\}\times\Sigma^*$ for the
set of channel actions over $\Sigma$, and often omit the $\Sigma$
subscript when it can be inferred from the context. We use
$\theta,\theta',\ldots$ to denote actions and $\sigma,\rho,\ldots$ to
denote sequences of channel actions.

We'll constantly refer to the \emph{written part} and the \emph{read
part} of some channel action (or sequence of such). These are formally
defined via
\begin{xalignat}{3}
\label{eq-def-wri-rea}
\begin{aligned}
\wri(!w) &\eqdef w              \:,
\\
\rea(!w) &\eqdef \epsilon       \:,
\end{aligned}
&&
\begin{aligned}
\wri(?w) &\eqdef \epsilon       \:,
\\
\rea(?w) &\eqdef w              \:,
\end{aligned}
&&
\begin{aligned}
\wri(\theta_1\cdots\theta_m) &\eqdef \wri(\theta_1) \cdots\wri(\theta_m) \:,
\\
\rea(\theta_1\cdots\theta_m) &\eqdef \rea(\theta_1) \cdots\rea(\theta_m) \:.
\end{aligned}
\end{xalignat}

\medskip

\noindent
\emph{Semantics.}
The operational semantics of LCMs is given via transition systems. Fix
some LCM $S=\tuple{Q,\Sigma,\Delta}$.  Actions in $\Act_\Sigma$ induce
a ternary relation $\step{}\subseteq
\Sigma^*\times\Act_\Sigma\times \Sigma^*$ on channel contents:
\begin{xalignat}{2}
\label{eq-def-sem-actions}
  x\step{!\:w}y
& \;\equivdef\; y \subword x \: w \:,
&
  x\step{?\:w}y
& \;\equivdef\; w \: y \subword x \:.
\end{xalignat}
Observe how \Cref{eq-def-sem-actions} includes the subword relation in
the definition of the operational semantics. This models the fact that
messages in the channel can be lost nondeterministically during any
single computation step. A consequence is the following monotonicity
property: if $x'\supword x$ and $y\supword y'$ then $x\step{\theta}y$
implies $x'\step{\theta}y'$.

A \emph{configuration} of $S$
is a pair $c=(q,x)\in Q\times\Sigma^*$ that denotes a current
situation where the control
of $S$ is set at  $q$ while the contents of the channel is $x$. We let
$\Conf_S\eqdef Q\times\Sigma^*$ denote the set of configurations.
The set of rules $\Delta$ induces a labelled transition relation
$\step{}\subseteq\Conf_S\times\Delta\times\Conf_s$
between configurations defined by
\begin{equation}
\label{eq-def-sem-lcms}
(q,x)\step{\delta}(q',y)
\;\equivdef\;
\delta\in\Delta \text{ has the form } \tuple{q,\theta,q'} \text{ and } x\step{\theta}y
\:.
\end{equation}
Several convenient notations are derived from the main transition
relation: we write $c\step{\theta}c'$ when $c\step{\delta}c'$ for a
rule $\delta$ that carries action $\theta$. When
$\sigma=\theta_1\theta_2\cdots\theta_m$ is a sequence of actions, we
write $c \step{\sigma} c'$ when there is a sequence of steps $c_0
\step{\theta_1}c_1 \step{\theta_2} c_2 \cdots \step{\theta_m} c_m$
with $c_0=c$ and $c_m=c'$.  Then $c\step{*}c'$ means that $c
\step{\sigma} c'$ for some sequence $\sigma$. Similar notations, e.g.,
``$x\step{\theta_1\:\theta_2}y$'' or ``$x\step{\theta^*}y$'', are used
for channel contents. In fact, since we shall mostly consider fixed
paths, or paths of a fixed shape, we will usually concentrate on the
channel contents and leave the visited locations implicit.

\medskip

\noindent
\emph{Flat LCMs.}
An elementary cycle of length $m$ in a LCM $S$ is a non-empty set
$C=\{\tuple{p_i,\theta_i,q_i}~|i=1,\ldots,m\}$ of rules from $\Delta$
such that $q_m=p_1$ and $p_i=q_{i-1}$ when $2\leq i\leq m$, and such
that the $p_i$'s are all distinct. A cycle of length 1 is a
\emph{self-loop}. The set $\{p_1,\ldots,p_m\}$ is the set of locations
\emph{visited} by $C$.  Note that two distinct cycles may have the
same visited set if they use different transition rules.

We say that $S$ is \emph{flat} if no  control location
is visited by two different elementary cycles.
An extreme case of flat machines are the
machines having no cycles whatsoever, called \emph{acyclic}
machines.\footnote{In the finite-automata literature, ``acyclic automata''
sometimes allow self-loops.}

When $S$ is flat, there is (at most) one cycle around any location $q$
and we write $\sigma_q$ for the sequence of actions along this cycle,
making sure that $\sigma_q$ starts with the action leaving $q$ (so
that if $q,q'$ are two locations visited by the same cycle,
$\sigma_{q'}$ will be a cyclic shift of $\sigma_q$).  When there is no
cycle visiting $q$ we let $\sigma_q=\epsilon$ by convention.

\medskip

\noindent
\emph{$\NP$-hardness.}
It is known that reachability and other verification problems are
$\NP$-hard for (reliable) FIFO automata:
see~\cite[App.~C]{esparza2012} and~\cite{finkel2019}. We strengthen
these results in \Cref{app-NP-hardness} with the following theorems
that cover reliable and unreliable channels indifferently.

\begin{theorem}[Hardness for acyclic channel machines]
\label{thm-hard-acyclic}
Reachability, nontermination and unboundedness are $\NP$-hard for
\emph{acyclic} channel machines, with reliable or with unreliable
channels. Hardness already holds for a single channel and a binary
alphabet. It also holds for a unary alphabet (i.e., for acyclic
VASSes, reliable or lossy) provided one allows several channels (or
counters).
\end{theorem}

$\NP$-hardness for acyclic machines uses the nondeterminism allowed in
channel machines. It is thus interesting to consider
\emph{single-path} machines where the control graph is a single line
possibly carrying cycles on some locations, as is done
in~\cite{kuhtz2011} or~\cite{demri2015}.  In such a machine,
nondeterminism only occurs in choosing how many times a cycle is
visited (and what messages are lost in unreliable systems). This is equivalent to considering reachability (or
nontermination or unboundedness) \emph{along a given bounded path scheme
of the form $q_1C_1^*q_2C_2^*\cdots q_mC_m^*$}.

\begin{theorem}[Hardness for single-path channel machines]
\label{thm-hard-single-path}
Reachability, nontermination and unboundedness are $\NP$-hard for
\emph{single-path} channel machines, with reliable or with unreliable
channels. Hardness already holds for a single channel. It also holds
for single-path VASSes, reliable or lossy, provided one allows several
counters.
\end{theorem}

The above $\NP$-hardness does not apply to bounded path schemes
\emph{with a fixed number of cycles} and indeed we show in
\Cref{sec-backward-reachability} that reachability along path schemes
with $m$ cycles can be verified in polynomial-time $n^{O(m)}$.

\section{Backward reachability in flat LCMs}
\label{sec-backward-reachability}

In this section we consider a generic flat single-channel LCM $S$ with
channel alphabet $\Sigma$ and investigate the complexity of
backward-reachability analysis.

\subsection{Computing predecessors}
\label{ssec-def-pr}

The classical approach to deciding reachability in LCMs is the
backward-reachability algorithm proposed by Abdulla and Jonsson. They
first developed it for lossy channel systems~\cite{abdulla96b}
before generalising it to the larger class of
Well-Structured Systems~\cite{abdulla2000c,finkel98b}.

For backward reachability, we write $\Pre[\sigma](x)$ for
$\{y\in\Sigma^*~|~y\step{\sigma}x\}$, the set of $\sigma$-predecessors
of $x$, and $\Pre[\sigma](\upc x)$ for $\{y\in\Sigma^*~|~\exists
x'\in\upc x:y\step{\sigma}x'\}$, the set of $\sigma$-predecessors of
$x$ ``and larger contents''. A consequence of the monotonicity of
steps is that $\Pre[\sigma](\upc x)$ is upward-closed set and, unless
$\sigma$ is the empty sequence, coincides with $\Pre[\sigma](x)$.

\begin{definition}[{$\pr[\sigma](x)$}]
\label{def-pr}
For a channel contents $x\in\Sigma^*$ and a sequence $\sigma$ of
channel actions, we write $\pr[\sigma](x)=y$ when $\Pre[\sigma](\upc
x)=\upc y$.
\end{definition}
In the case of lossy channels, $\Pre[\sigma](\upc x)$ always has a
single minimal element, hence $\pr[\sigma](x)$ is always defined.  We
now explain how to compute it.  \\

For two words $x,v$ we define $x/v$ as the prefix of $x$ that remains
when we remove from $x$ its longest suffix that is a subword of
$v$. This operation is always defined and can be computed using the
following rules where $a,b$ are letters:
\begin{xalignat}{3}
\label{eq-def-residual}
x / \epsilon &= x			\:,
&
\epsilon / v &= \epsilon		\:,
&
(x\:a) / (v\:b) &= \begin{cases} x / v    &\text{if $a=b$,}\\
			    (x\:a) / v & \text{if $a\neq b$}.
		\end{cases}
\end{xalignat}
This immediately entails $(x/v)/v' = x / (v' \: v)$. We'll also use
the following
properties:
\begin{xalignat}{2}
\label{eq-useful}
&\text{if } x'/v\neq\epsilon \text{ then } x(x'/v) = (x \: x')/v    \:,
&
&\text{if } |x'|>|v| \text{ then } x'/v \neq\epsilon                \:.
\end{xalignat}

We may now compute $\pr[\sigma](x)$ with:
\begin{equation}
\label{eq-pr-def}
\begin{aligned}
\pr[?u](x) &= u\cdot x				    \:,
&
\pr[!v](x) &= x/v				    \:,
\\
\pr[\epsilon](x) &= x						\:,
&
\;\;\;\;\;\;\;\;\; 
\pr[\sigma_1\cdot\sigma_2](x) &= \pr[\sigma_1]\bigl(\pr[\sigma_2](x)\bigr) \:.
\end{aligned}
\end{equation}

W.r.t.\ subword ordering, the $/$ operation is monotonic in its first
argument and contramonotonic in the second : $u\subword u'$ implies
$u/v\subword u'/v$ and $x/u\supword x/u'$.  Concatenation too is
monotonic. This generalises to the following useful lemma:
\begin{lemma}
\label{lem-pre-mono}
Assume $\pr[\sigma](x)=y$ and $\pr[\sigma](x')= y'$
where $\sigma$ is some sequence of actions. Then $x\subword x'$ implies
$y\subword y'$.
\end{lemma}
\begin{proof}
By induction on $\sigma$, using \cref{eq-def-residual,eq-pr-def}.
\end{proof}

\subsection{Cycles: repeating a given sequence of actions}

We now focus on computing $\pr[\sigma^k](x)$ for $\sigma$ a sequence
of actions and some $k\in\Nat$.

Without any loss of generality, $\sigma$ can be written in the general
form $?a_1\:!b_1\:?a_2\:!b_2\cdots?a_r\:!b_r$ where each $a_i$ and
$b_i$ is a letter or the empty word $\epsilon$.  Then
$\rea(\sigma)=a_1a_2\cdots a_r$ and $\wri(\sigma)=b_1b_2\cdots b_r$.

To fix notation, we define ``\emph{the small-step sequence for
$\pr[\sigma](x)$}'', or just ``\emph{the SSS}'', as the sequence $y_r$, $y'_{r}$, $y_{r-1}$,
$y'_{r-1}$, $\ldots$, $y_1$, $y'_1$, $y_0$ of $2r+1$ words given by
\begin{xalignat}{3}
\label{eq-compute-yi}
y_r &= x		 \:,
&
y'_i &= y_i/b_i	      \:,
&
y_{i-1} &= a_i \: y'_{i}	   \:.
\end{xalignat}
Clearly, the SSS lists all the intermediary steps
in the computation of $\pr[\sigma](x)$ as dictated by
\cref{eq-pr-def}, and thus
it yields $y_0=\pr[\sigma](x)$.
\\

Our first lemma handles the special case where $x$ is made of copies of $\rea(\sigma)$.
\begin{lemma}
\label{lem-up-sigma-um}
Let $u=\rea(\sigma)$.\\
(i) If $x$ is a fractional power $u^p$ of $u$, then $y=\pr[\sigma](x)$
is also a fractional power of $u$, written $y=u^{m}$. \\
(ii) Furthermore, if $m>1$, then $\pr[\sigma](u^{p+n})=u^{m+n}$ for
all $n\in\Nat$.\\
(iii) Finally, for all $n\in\Nat$, if $m>n+1$, then $\pr[\sigma](u^{p-n})=u^{m-n}$.
\end{lemma}
\begin{proof}
The lemma holds spuriously if $u=\epsilon$, so we assume $|u|>0$.  Let
us write $\sigma$ in the general form
$?a_1\:!b_1\:?a_2\:!b_2\cdots?a_r\:!b_r$, so that $u=a_1a_2\cdots
a_r$.  To simplify notation we will write $u_{(i)}$ for the shift
$a_{i+1}\cdots a_r\cdot a_1 \ldots a_i$ that really should be written
$u_{(|a_1\cdots a_i|)}$ (remember that $a_j=\epsilon$ is possible).

We now claim that, in the SSS $(y_i,y'_i)_i$ for $\sigma$ and
$x$, each $y_i$ and $y'_i$ is a fractional power of $u_{(i)}$, written
$y_i=u_{(i)}^{p_i}$ and $y'_i=u_{(i)}^{p'_i}$.

The proof is by induction on $r-i$.  For $y_i$, there are two cases:
(1) $y_r=x$ is a power of $u$ by assumption, hence of $u_{(r)}$, with
$p_r=p$; (2) $y_{i-1}$ is $a_i \: y'_i$, i.e., $a_i\: u_{(i)}^{p'_i}$
by ind.\ hyp., hence a power of $u_{({i-1})}$ with
$p_{i-1}=p'_i+\frac{|a_i|}{|u|}$. For $y'_i$ the proof is simpler: by
ind.\ hyp.\ it is $u_{(i)}^{p_i}/b_i$ and, as a prefix of a power of
$u_{(i)}$, is itself a power of $u_{(i)}$, albeit with a perhaps
smaller exponent, i.e., $p_i-\frac{|b_i|}{|u|}\leq p'_i\leq p_i$.  \\

\noindent
(i) Since $y$ coincide with $y_0$, we obtain $y=u^m$ as required by
letting $m=p_0$.  \\

\noindent
(ii) \Cref{eq-pi-inequalities} gathers the (in)equalities we just
established:
\begin{xalignat}{4}
\label{eq-pi-inequalities}
p_r &= p\:,
&
p_{i-1} &= p'_i+\frac{|a_i|}{|u|}\:,
&
\max \Bigl(0,p_i-\frac{|b_i|}{|u|}\Bigr) &\leq p'_i\leq p_i\:,
&
p_0 &= m        \:.
\end{xalignat}
Thus the assumption $m>1$ entails $p'_i>0$, i.e.\ $y'_i\neq\epsilon$,
for all $i=r,r-1,\ldots,2,1$.  Let us now consider the SSS
$(z_i,z'_i)_i$ for $\pr[\sigma](u^{p+1})$. We claim that for all $i$,
$z_i=u_{(i)}y_i$ and $z'_i=u_{(i)}y'_i$, as is easily proven by
induction on $r-i$. The crucial case is $z'_i$, defined as $z_i/b_i$
and equal to $(u_{(i)}y_i)/b_i$ by ind.\ hyp.  Since
$y_i/b_i=y'_i\neq\epsilon$ as just observed, we deduce
$(u_{(i)}y_i)/b_i=u_{(i)}(y_i/b_i)$ from \cref{eq-useful}.  This is
$u_{(i)}y'_i$ as required.  Finally we end up with
$\pr[\sigma](u^{p+1})=z_0=u_{(0)}y_0=u\: u^m=u^{m+1}$, and this
generalises to $\pr[\sigma](u^{p+n})=u^{m+n}$.  \\

\noindent
(iii) With \cref{eq-pi-inequalities}, the assumption $m>n+1$ now
 entails $p_i,p'_i\geq n+\frac{1}{|u|}$ for all $i$. We claim that the SSS
$(z_i,z'_i)_i$ for $\pr[\sigma](u^{p-n})$ satisfies $u_{(i)}^nz_i=y_i$
and $u_{(i)}^nz'_i = y'_i$ for all $i$, as can be proved by induction on
$r-i$. The base case $u^n z_r=u^n u^{p-n}=u^p=y_r$ is clear.  Let us now
consider $u_{(i)}^n z'_i$. It is $u_{(i)}^n(z_i/b_i)$, that is
$u^n_{(i)}(u_{(i)}^{p_i-n}/b_i)$ since $u_{(i)}^n z_i=y_i$ by ind.\ hyp.\
and $y_i=u_{(i)}^{p_i}$ by (i).
 Now $|u_{(i)}^{p_i-n}| =
|u|(p_i-n) \geq 1 \geq |b_i|$, so \cref{eq-useful} applies and we deduce
$u_{(i)}^n(z_i/b_i) = (u_{(i)}^n z_i)/b_i$ $=y_i/b_i$ (by ind.\ hyp.)
$=y'_i$. We have proved $u_{(i)}^n z'_i=y'_i$ as required.
Finally, proving $u_{(i)}^n z_i=y_i$ is handled in a similar way.

\end{proof}
Note that $m>1$ is required for part (ii) of the Lemma.	 For example,
with $\sigma=\:?a\:!b\:?c\:!c\:!a$ one has $u=\rea(\sigma)=a\:c$ and
$\pr[\sigma](u^{\frac{1}{2}})=u^1$. However one can check that
$\pr[\sigma](u^{\frac{3}{2}}) = a\:c\:a = u^{\frac{3}{2}}$.
\\

Equipped with \Cref{lem-up-sigma-um}, we turn to the general case for $\pr[\sigma^k](x)$.
\begin{theorem}
\label{thm-yk-pk-lk}
Let $\sigma\in\Act_\Sigma^*$ be a sequence of actions and write $u$ for
$\rea(\sigma)$. Let $x\in\Sigma^*$ be some channel contents and
write $y_k$ for $\pr[\sigma^k](x)$.  \\
(i) For every $k\in \Nat$, $y_k$ has the form $u^{p_k}\cdot
x_{<\ell_k}$ for some fractional power $p_k$ and some length $\ell_k\in
\{0,1,\ldots,|x|\}$.\\
(ii) Furthermore, computing $p_k$ and $\ell_k$ can be done in time
 $\poly(|\sigma|+|x|+\log k)$.
\end{theorem}
\begin{proof}
(i) Write $v$ for $\wri(\sigma)$ and consider the sequence
$(x_k)_{k\in\Nat}$ given by $x_0\eqdef x$ and $x_{k+1}\eqdef
x_k/v$. Note that $|x_{k+1}|\leq|x_k|$ for all $k$ and write
$\bmkappa$ for the largest index with $x_\bmkappa\neq\epsilon$. We let
$\bmkappa=-1$ if already we started with $x=\epsilon$, and
$\bmkappa=\omega$ if all $x_k$'s are non-empty, which happens iff
$\alphabet(x)\not\subseteq\alphabet(v)$.

If $k\leq \bmkappa$, $\pr[\sigma^k](x)=u^k\cdot x_k$ and $x_k$ is a
prefix of $x$, so taking $p_k=k$ and $\ell_k=|x_k|$ works.

If $k=\bmkappa+1$, $y_k$ is $\pr[\sigma](u^\bmkappa\cdot
x_{<\ell_\bmkappa})$. Since $x_{<\ell_\bmkappa}/v=\epsilon$, the
result is a prefix of $u^{\bmkappa+1}$, so has the form
$u^{p_{\bmkappa+1}}$ for some $p_{\bmkappa+1}$. One also lets
$l_{\bmkappa+1}=0$.

Finally, if $k>\bmkappa+1$, we have $y_k =
\pr[\sigma^{k-\bmkappa-1}](y_{\bmkappa+1}) =
\pr[\sigma^{k-\bmkappa-1}](u^{p_{\bmkappa+1}})$ and we just have to
invoke \Cref{lem-up-sigma-um} (and set $l_k=0$).  \\

\noindent
(ii) Computing $\bmkappa$ takes time $O(|x|+|\sigma|)$.

If $k\leq \bmkappa$, comparing $k$ with $\bmkappa$ and computing $p_k$
and $\ell_k$ takes additional time $O(|x|+|\sigma|+\log k)$.

If $k=\bmkappa+1$, we need to compute
$\pr[\sigma](u^\bmkappa\:x_{<\bmkappa})$ in order to extract
$p_{\bmkappa+1}$.  This uses \cref{eq-pr-def} for
$O(|\sigma|)$ small steps. Note that we do not build $u^\bmkappa$
explicitly: once $x$ has been consumed, we work on some $u_{(i)}^{p}$
and just update $p$ and $i$ when applying some $\pr[?a]$, or only
update $p$
when applying some $\pr[!b]$, for which we only need to know where are
the occurrences of $b$ in $u$. For each small step, the updates can be
computed in time $O(|u|+|x|)$, hence $p_{\bmkappa+1}$ is computable in
quadratic time.

If $k>\bmkappa+1$, we set $q_0=p_{\bmkappa+1}$, $k'=k-\bmkappa-1$ and
aim for $\pr\bigl[\sigma^{k'}\bigr](x')$, starting from $x'=u^{q_0}$.
We need to compute $u^{q_{k'}}$ in the sequence
$u^{q_0},u^{q_1},\ldots,$ defined by $u^{q_{i+1}} \eqdef
\pr[\sigma](u^{q_i})$.  Let us first compute $q_1$ and consider the
three possibilities:\\
(1) If $q_0 = q_1$, $y_p$ is a fixpoint for $\pr[\sigma]$ and we know
$p_k=p$.  \\
(2) If $q_0 < q_1$, the exponents increase under $\pr[\sigma]$ and
after computing at most $|u|+1$ consecutive values, we'll find two
indexes $1\leq i<j\leq|u|+1$ such that $q_i$ and $q_j$ have the same
fractional parts, i.e., differ by some natural number.  We can then
use \Cref{lem-up-sigma-um}.(ii) and compute
$q_{j+\left\lfloor\frac{k'-j}{j-i}\right\rfloor} = q_j + \left\lfloor\frac{k'-j}{j-i}\right\rfloor$.
From there, we're just at most $|u|$ steps from $q_{k'}$, i.e., $p_k$.
\\
(3) Finally, if $q_0>q_1$ a similar technique, now relying on
\Cref{lem-up-sigma-um}.(iii), will let us compute $q_{k'}$ in
polynomial time.
\end{proof}

The next step is to compute $\Pre[\sigma^*](\upc x)$, that is, $\upc
x\cup\Pre[\sigma](\upc x)\cup \Pre[\sigma^2](\upc x)\cup \cdots$.
Like $\Pre[\sigma](\upc x)$, this set is upward-closed. However it may
have several minimal elements and one needs to collect all of them in
order to represent the set faithfully.

\begin{definition}[Iteration number]
\label{def-L-sigma-x}
The \emph{iteration number} $L(\sigma,x)$ associated with a sequence
of actions $\sigma$ and a channel contents $x$ is the smallest integer
such that there exists $\ell\leq L(\sigma,x)$ with
$\pr[\sigma^\ell][x]\subword \pr[\sigma^{L(\sigma,x)+1}](x)$.	Note
that, by Higman's Lemma, such an integer always exists.
\end{definition}

The point of \Cref{def-L-sigma-x} is that it captures the number of
iterations that are sufficient to compute $\Pre[\sigma^*](\upc x)$.
\begin{lemma}
\label{lem-L-isfixpoint}
$\Pre[\sigma^*](\upc x) = \bigcup_{i=0}^{L(\sigma,x)}\upc
\pr[\sigma^i](x).$
\end{lemma}
\begin{proof}
Write $y_k$ for $\pr[\sigma^k](x)$ and $L$ for $L(\sigma,x)$. By
definition there is some $\ell\leq L$ with $y_\ell\subword y_{L+1}$.
By \Cref{lem-pre-mono}, this continues into $y_{\ell+1}\subword
y_{L+2}$, $y_{\ell+2}\subword y_{L+3}$, etc., implying $\upc
y_\ell\supseteq \upc y_{L+1}$, $\upc y_{\ell+1}\supseteq \upc
y_{L+2}$, $\upc y_{\ell+2}\supseteq \upc y_{L+3}$, \ldots\ Finally
$\Pre[\sigma^*](\upc x)$, which is $\bigcup_{i\in\Nat}\upc y_i$
coincides with the finite union $\bigcup_{i=0}^{L(\sigma,x)}\upc y_i$.
\end{proof}

\begin{theorem}[Bounding iteration numbers]
\label{thm-L-sigma-x}
$L(\sigma,x)\leq |x|(|\rea(\sigma)|+1)$
for any action sequence $\sigma$ and channel contents $x$.
\end{theorem}
\begin{proof}
We write $u$ for $\rea(\sigma)$.  Using \Cref{thm-yk-pk-lk}, we write
$y_k=\pr[\sigma^k](x)=u^{p_k}\cdot x_{<\ell_k}$ and observe that
$p_i\leq p_j$ and $\ell_i\leq\ell_j$ imply $y_i\subword y_j$.  Recall
from the proof of \Cref{thm-yk-pk-lk} that $|x|=\ell_0\geq \ell_1\geq
\cdots\geq\ell_i\geq \cdots$ is a decreasing sequence and that $p_k=k$
when $\ell_k>0$.

There are two cases:\\
(1) If $(\ell_k)_k$ stabilises with some limit value $\ell_\infty$ that is
strictly positive, then $\ell_{|x|-1}=\ell_{|x|}$ and we deduce
$y_{|x|-1}\subword y_{|x|}$, entailing $L(\sigma,x)< |x|$.

\noindent
(2) If $\ell_\infty=0$ then, writing $k_0$ for the first index with
$\ell_{k_0}=0$, we know that $k_0\leq |x|$ and $p_{k_0}=k_0$.  If
$p_{k_0+1}\geq p_{k_0}$ then $y_{k_0}\subword y_{k_0+1}$.  Otherwise
$p_{k_0}>p_{k_0+1}$ and as a consequence of \Cref{lem-pre-mono} the
suffix sequence $p_{k_0}> p_{k_0+1}\geq p_{k_o+2}\geq p_{k_0+3} \geq
\cdots$ is decreasing.	Since the $p_k$ fractions are multiples of
$\frac{1}{|u|}$, the sequence $(p_k)_{k\geq k_0}$ can only take
$1+|u|k_0$ different values and eventually yield $p_k=p_{k+1}$ for
some $k \leq k_0+k_0|u|\leq |x|(|u|+1)$, entailing $L(\sigma,x)\leq
|x|(|u|+1)$ as claimed.
\end{proof}

The bound given by \Cref{thm-L-sigma-x} is tight as
the next simple example shows.
\begin{example}[Bounds for $L(\sigma,x)$ are tight.]
\label{ex-quadratic-L}
For $\sigma=\:!ab^5\:?b^4$ and $x=a^4$, the $(y_k)_k$ sequence with
$y_k \eqdef \pr[\sigma^k](x)$ is:
\begin{center}
\begin{tikzpicture}[algntxt/.append style={font=\vphantom{Ag}},
		    auto,node distance=2.7em] 
{
\node [algntxt,label=below:{\footnotesize $y_0$}] (y0) {$a^4$,};
\node [algntxt,right of=y0,label=below:{\footnotesize $y_1$}] (y1) {$b^4a^3$,};
\node [algntxt,right of=y1,label=below:{\footnotesize $y_2$}]  (y2) {$b^8a^2$,};
\node [algntxt,right of=y2,label=below:{\footnotesize $y_3$}]  (y3) {$b^{12}a$,};
\node [algntxt,right of=y3,label=below:{\footnotesize $y_4$}]  (y4) {$ b^{16}$,};
\node [algntxt,right of=y4,label=below:{\footnotesize $y_5$}]  (y5) {$b^{15}$,};
\node [algntxt,right of=y5,label=below:{\footnotesize $y_6$}]  (y6) {$b^{14}$,};
\node [algntxt,right of=y6,label=below:{\footnotesize $y_7$}]  (y7) {$b^{13}$,};

\node [algntxt,right=4em of y7,label=below:{\footnotesize $y_{19}$}]	 (y19) {$b$,};
\node [algntxt,right of=y19,label=below:{\footnotesize $y_{20}$}]  (y20) {$\epsilon$,};
\node [algntxt,right of=y20,label=below:{\footnotesize $y_{21}$}]  (y21) {$\epsilon$,};
\node [algntxt,right of=y21]  (ydots) {$\ldots$};
\node [algntxt] at ($(y7)!0.5!(y19)$) (mdots) {$\ldots$};

\node[yshift=1.4em] at ($(y0)!0.5!(y4)$) (br1) {$\obracew{\hspace{12.3em}}{|x|=l_0>l_1>\cdots>l_4=0}$};
\node[yshift=1.6em] at ($(y4)!0.5!(ydots)$) (br2) {$\obracew{\hspace{24em}}{k_0=4=p_{k_0} \land p_4\geq p_5\geq p_6 \cdots}$};
}
\end{tikzpicture}
 \end{center}
Since $y_{20}\subword y_{21}$ is the earliest increasing pair,
\Cref{def-L-sigma-x} gives $L(!ab^5\:?b^4,a^4)=20$.

This generalises to $L(!ab^{n+1}\:?b^n,a^{m})=m(n+1)$ for any
$n,m\in\Nat$, which is exactly the $|x|(|\rea(\sigma)|+1)$ bound given
by \Cref{thm-L-sigma-x}.
\qed
\end{example}

\subsection{Bounding runs}
\label{sec-bounding-runs}

Assume that a flat LCM $S$ is such that $(q',y)\in\post^*(q,x)$. Since $S$ is flat, the run
$(q,x)\step{*}(q',y)$ has the following shape:
\begin{equation}
\label{eq-run-shape}
\begin{gathered}
(q,x)=(q_0,z_0)
\step{\sigma_0^{n_0}}
(q_0,z'_0)
\step{\theta_1}
(q_1,z_1)
\step{\sigma_1^{n_1}}
(q_1,z'_1)
\step{\theta_2}
(q_2,z_2)
\step{\sigma_2^{n_2}}
\;\cdots
\\
\;\;\;\;\;\;\;\;\;\;\;\;\;\;\;\;\;\;\;\;\;\;\;\;\;\;\;\;\;\;\;\;
\cdots\;\;
(q_{m-1},z'_{m-1})
\step{\theta_m}
(q_m,z_m)
\step{\sigma_m^{n_m}}
(q_m,z'_m) = (q',y)\:.
\end{gathered}
\end{equation}
In \cref{eq-run-shape}, the control locations $(q=)q_0,q_1,\ldots,q_m(=q')$
are all distinct, $\sigma_i$ is the sequence of actions performed
along the (unique) cycle on $q_i$, and $n_i$ is the number of times
this cycle has been traversed along the run. We use $\sigma_i=\epsilon$
when there is no cycle on $q_i$, and we use $n_i=0$ when the cycle is
not traversed at all. For $i=1,\ldots,m$, $\theta_i$ is the sequence of
actions that labels the transition from $q_{i-1}$ to $q_i$. \\

We say that the run in \cref{eq-run-shape} is \emph{minimal} if for
all $i=1,\ldots,m$, $z_i$ is a minimal element in
$\Pre[\sigma_i^*](\upc z'_i)$ and $n_i$ is the smallest such
$z_i=\pr[\sigma_i^{n_i}](z'_i)$, and if
$z_i=\pr[\sigma_i^{n_i}](z'_i)$ for $i=0,\ldots,m$. By allowing
$z_0\subword x$, it is always possible to associate a minimal run
with some reachability statement ``$(q,x)\step{*}(q',y)$'' and
use the tuple
\begin{equation}
\label{eq-witness}
\langle q_0, z_0, n_0, z'_0, q_1, z_1, n_1, z'_1, \ldots, q_m, z_m,
n_m, y \rangle
\end{equation}
as a witness of reachability.

We now try to bound the size of such a witness.  One has
\begin{xalignat}{3}
|z'_m|&=|y|				\:,
&
|z_i|&\leq |z'_i| + n_i |\rea(\sigma_i)|	\:,
&
|z'_{i-1}|&\leq |z_i| + |\rea(\theta_i)|	\:,
\end{xalignat}
for all $i$.  We further know from \Cref{thm-L-sigma-x}, that $n_i\leq
|z'_i|(1+|\rea(\sigma_i)|)$ for $i=0,\ldots,m$.

Thus, writing $n$ for the size $|S|+|x|+|y|$ of the instance (so that
$m\leq n$, and $|\rea(\sigma_i)|,|\rea(\theta_i)|\leq n$ for all
$i$), we have quadratic bounds $O(n^2)$ for $n_m$ and $|z_m|$, cubic
bounds $O(n^3)$ for $n_{m-1}$ and $|z_{m-1}|$, \ldots, etc., so that
the witness has size $O(n^m)$, hence $2^{O(n)}$.  \\

Unfortunately, as \Cref{ex-exponential-flatLCM} shows, these bounds
cannot be much improved upon.

\begin{example}
\label{ex-exponential-flatLCM}
Consider the flat LCM $S$ depicted in \cref{fig-exp-LCM}.
In $S$,
$(q_0,\epsilon)\step{*}(q'_0,\epsilon)$ is witnessed by the following run schema
\begin{equation*}
\begin{split}
(q_0,\epsilon)
&\step{} (q_1,a b)
\step{*}	 (q_1,b a^2)
\step{}	 (q_2,a^2 b)
\step{*}	 (q_2,b a^4)
\step{}	 (q_3,a^4 b)
\step{*}	 \cdots
\step{*}	 (q_n,b a^{2^n})
\\
&\step{}	  (q'_n,a^{2^n} b)
\step{*} (q'_n,b a^{2^{n-1}})
\step{}	 (q'_{n-1},a^{2^{n-1}} b)
\step{*}	 \cdots
\step{}	 (q'_1,a^2 b)
\step{*}	 (q'_1,b a)
\step{}	 (q'_0,\epsilon)
\:.
\end{split}
\end{equation*}
\begin{figure}[htbp]
\centering
\scalebox{0.75}{
  \begin{tikzpicture}[->,>=stealth',shorten >=1pt,node distance=6em,thick,auto,bend angle=30]
\tikzstyle{every state}=[minimum size=2.3em]

\node at (-2,0)	  [state] (cm1)     {\footnotesize $q_0$};
\node at (0,0)	  [state] (c0)      {\footnotesize $q_1$};
\node at (2,0)	  [state] (c1)      {\footnotesize $q_2$};
\node at (4,0)	  [state] (c2)      {\footnotesize $q_3$};
\node at (5.2,0)  []	  (cdots)   {\Large $\cdots$};
\node at (7,0)	  [state] (c3)      {\footnotesize $q_n$};
\node at (9,0)	  [state] (c4)      {\footnotesize $q'_n$};
\node at (11,0)	  [state] (c5)      {};
\node at (11,0)	  []      (c5label) {\footnotesize $q'_{n-1}$};
\node at (12.2,0) []	  (cdots2)  {\Large $\cdots$};
\node at (14,0)	  [state] (c6)      {\footnotesize $q'_1$};
\node at (16,0)	  [state] (c7)      {\footnotesize $q'_0$};

\path (c0) edge [loop above] node {$?a  \: !aa$} (c0);
\path (c1) edge [loop above] node {$?a  \: !aa$} (c1);
\path (c2) edge [loop above] node {$?a  \: !aa$} (c2);
\path (c3) edge [loop above] node {$?a  \: !aa$} (c3);
\path (c4) edge [loop above] node {$?aa \: !a$} (c4);
\path (c5) edge [loop above] node {$?aa \: !a$} (c5);
\path (c6) edge [loop above] node {$?aa \: !a$} (c6);

\path (cm1) edge  node  {$!ab$}    (c0);
\path (c0)  edge  node  {$?b\:!b$} (c1);
\path (c1)  edge  node  {$?b\:!b$} (c2);
\path (6,0) edge [dashed] node {} (c3);
\path (c3)  edge  node  {$?b\:!b$} (c4);
\path (c4)  edge  node  {$?b\:!b$} (c5);
\path (13,0) edge [dashed]  node {} (c6);
\path (c6)  edge  node  {$?ba$}    (c7);

  \end{tikzpicture}
}
\caption{A flat LCM where $(q_0,\epsilon)\step{*}(q'_0,\epsilon)$ requires exponential-sized configurations.}
\label{fig-exp-LCM}
\end{figure}
 In fact, there is only one run witnessing $(q_0,\epsilon)
\step{*} (q'_0,\epsilon)$ and this run necessarily visits
$(q_n,b a^{2^n})$, a configuration of exponential size, iterating
$2^{n-1}$ times the cycle on $q_n$.  Observe that, starting from
$(q_0,\epsilon)$, any message loss will prevent ever reaching $q'_0$.
\qed
\end{example}

\section{SLP-compressed words and an $\NP$ algorithm for reachability}
\label{sec-NP-via-SLP}

In this section we explain how the exponentially long minimal runs
analysed in \Cref{sec-bounding-runs} can be handled efficiently using
SLP-compressed words. This provides witnesses of polynomial size
that can be validated in polynomial time, thus showing that reachability
in flat LCMs is in $\NP$.

\subsection{SLP-compressed words}
\label{ssec-SLP}

\emph{Compressed words} are data structures used to represent long
words via succinct encodings. If a long word is rather repetitive, it
can have a succinct encoding of logarithmic size.  Since several
operations on long words or decision tests about them can be performed
efficiently on the succinct representation, compressed words have been
used to provide efficient solutions to algorithmic problems involving
exponential-size (but rather repetitive) words, see~\cite{lohrey2012}
for a survey.

The most studied encoding is the SLP, for Straight-Line Program, which
is in effect an acyclic context-free grammar that generates a single
word, called its \emph{expansion}.

From now on, we always use small letters $x,y,u,v$ for usual words,
and capital letters $X,Y,U,V$ for SLPs expanding to the corresponding
words.	Since SLPs are interpreted as plain words, we will
use them freely in places where words can be used. It will always be
clear when we consider the SLP as a data structure and then we use it
to denote its expansion. The main situation where we want to
distinguish between the two usages is when reasoning about size and
algorithmic complexity: for this we
write $|X|$ for the length $|x|$ of the expansion, while we write
$\size{X}$ for the size of the SLP as a data structure.	 For example,
if $X$ expands to $x$ then for any fractional power of the form $x^p$, there is an SLP $X^p$ with
$|X^p|=|x^p|=p|x|$ and $\size{X^p}=O(\size{X}+\log p)$.

In the rest of this section we will use well-known, or easy to prove,
algorithmic results on SLP. In particular, all the following problems
can be solved in polynomial time:
\begin{description}
\item[length:] Given a SLP $X$, compute $|X|$.

\item[factor:] Given a SLP $X$ and two positions $0\leq i\leq
  j\leq|X|$, construct a SLP of size $O(\size{X})$ for the factor $X[i:j]$.

\item[concatenation:] Given two SLPs $X$ and $Y$, construct a SLP for
  $X\cdot Y$.

\item[matching:] Given two SLPs $X$ and $Y$, decide if $X$ is a factor
  (or a prefix, or a suffix) of $Y$.
\end{description}

\medskip

To this list we add results tailored to our needs:
\begin{description}
\item[(scattered) subword with a power word:] Given a SLP $X$, a plain
  word $v$ and some power $k\in\Nat$, decide if $X\subword v^k$.  This
  special case of the fully compressed subsequence test can be done in
  time $\poly(\size{X}+|v|+\log k)$, see \Cref{prop-vn-subword-X} in
  the Appendix.

\item[iterated LCM predecessor:] Given a SLP $X$, a plain word $v$,
and some power $k\in\Nat$, compute a SLP for $X/v^k$, i.e., for
        $\pr[(!v)^k](X)$. This can be done
in time $\poly(\size{X}+|v|+\log k)$, see \Cref{prop-X-remove-vn} in
the Appendix.
\end{description}

With the above results, we are ready to lift the computation of
$\pr[\sigma^k](x)$ from plain words to SLPs:
\begin{proposition}
\label{prop-pre-sigma-SLP}
Given an SLP $X$, a sequence of actions $\sigma$, and some $k\in\Nat$,
it possible to compute an SLP $Y$ for $\pr[\sigma^k](X)$ in time
$\poly(\size{X}+|\sigma|+\log k)$.
\end{proposition}
\begin{proof}[Proof (sketch)]
We follow the construction described in the proof of
\Cref{thm-yk-pk-lk}, now using SLPs.  So again let us write $u$ and
$v$ for $\rea(\sigma)$ and $\wri(\sigma)$.

The first step is to compute $\bmkappa$.  This is done by dichotomic
search, since we can decide in polynomial time whether a candidate $n$
leads to $X/v^n=\epsilon$. We then build $X_{<\ell_\bmkappa}$ as $X/v^{\bmkappa}$.

If $k\leq \bmkappa$, we build a SLP $Y$ for $u^k\cdot
(X/v^k)$ and we are done.

If $k\geq\bmkappa+1$, we compute a SLP for $y_{\bmkappa+1}=u^{p_{\bmkappa+1}}$ by
applying $\pr[\sigma]$ on a SLP for $y_\bmkappa=u^\bmkappa\cdot x_{<\ell
\bmkappa}$: this involves computing a SSS involving at most $2m$
operations like prefixing by $a_i$ or computing $Y/b_j$. This is done
in polynomial time and the exponent in $u^{p_{\bmkappa+1}}$ can be
computed by dividing the length of a SLP with the length of $u$.
From there we continue as in the proof of \Cref{thm-yk-pk-lk}. This involves
performing a polynomial number of simple $\pr$ operations and some
simple reasoning on the exponents.
\end{proof}

\subsection{Reachability for flat LCMs is in \NP}
\label{ssec-NP-algo}

We now explain how \cref{eq-witness} can be replaced by an SLP-based
witness of the form
\begin{equation}
\tag{\ref{eq-witness}'}
\label{eq-witness-slp}
\langle q_0, Z_0, n_0, Z'_0, q_1, Z_1, n_1, Z'_1, \ldots, q_m, Z_m,
n_m, Y \rangle \:.
\end{equation}

\begin{lemma}
\label{lem-SLP-witness-size}
If $\langle q_0, z_0, n_0, z'_0, q_1, n_1, z_1, z'_1, \ldots, q_m,
z_m, n_m, y \rangle$ is a minimal witness for
$(q,x)\step{*}(q',y)$ in $S$, then there exist SLPs
$Z_0,Z'_0,Z_1,\ldots,Z_m,Y$ representing $z_0,z'_0,z_1,\ldots,z_m,y$
that have size polynomial in $|S|+|y|$.
\end{lemma}
\begin{proof}
By induction on $m-i$. We start with $Y$ for $y$ which does not need
any compression (and let $Z'_m=Y$ for the inductive reasoning).

Then any $Z_i$ has the shape $U_i^{p_i}\cdot (Z'_i)_{<\ell_i}$ for
some $p_i$ and $\ell_i$.  Now $\size{(Z'_i)_{<\ell_i}}$ is in
$O(\size{Z'_i})$ and since $p_i$ is in $2^{O(|S|)}$ ---as shown in
\Cref{sec-bounding-runs}---, the size of the SLP for $u_i^{p_i}$ is is
$O(|u_i|+|S|)$, i.e., $O(|S|)$.

Now any $Z'_{i-1}$ is $\pr[\theta_i](Z_i)$ and is easily obtained
from $Z_i$ and $\theta_i$ according to \cref{eq-pr-def}.
One can ensure that $\size{Z'_i}$ is in $O(\size{Z_i}+|S|)$.

Finally, and since each SLP has size linearly bounded in the size of
the following one (the bounds propagate from right to left), we have a
quadratic bound on the individual sizes for the $Z_i$ and $Z'_i$,
hence a cubic bound on the SLP witness overall (recall that the $n_i$,
written in binary, have size $O(|S|)$).
\end{proof}

\begin{theorem}
\label{thm-reach-in-NP}
Deciding whether $(q,x)\step{*}(q',y)$ in a flat LCM $S$ is
$\NP$-complete.
\end{theorem}
\begin{proof}
$\NP$-hardness is proven in \Cref{app-NP-hardness} and we just
provide a $\NP$ decision algorithm.

As expected, the algorithm just guesses a SLP-based witness and checks
that it is indeed a valid witness. For a positive instance of the
problem, a witness exists and has polynomial size as shown in
\Cref{lem-SLP-witness-size}. Now checking that it is valid, i.e., that
each $Z_i$ is indeed $\pr[\sigma^{n_i}](Z'_i)$ etc., can be done in
polynomial time as shown with \Cref{prop-pre-sigma-SLP}.\footnote{In
fact, it is sufficient to guess the exponents $n_1,\ldots,n_m$ for the
$\sigma_i$'s since the $Z_i,Z'_i$'s can be computed from them.}
\end{proof}

\section{$\NP$ algorithms for liveness properties}
\label{sec-nontermination-and-unboundedness}

We show in this section how, for flat LCMs, liveness properties like
nontermination, unboundedness, and existence of a B\"{u}uchi run,
effectively reduce to reachability. This only requires characterising
and computing the set of configurations from which infinite runs are
possible but \Cref{sec-backward-reachability} provides all the
necessary tools.  \\

With any sequence of channel actions $\sigma$ we associate
$I_\sigma\eqdef\bigcap_{k=0,1,2,\ldots}\Pre[\sigma^k](\Sigma^*)$.

\begin{lemma}
\label{lem-Isigma-principal}
$I_\sigma\subseteq\Sigma^*$ is an upward-closed set of channel
contents. It has a single minimal element or is empty.
\end{lemma}
\begin{proof}
Write $(y_k)_{k\in\Nat}$ for the sequence $y_0\eqdef\epsilon$ and
$y_{k+1}=\pr[\sigma](y_k)$.  Then $\Pre[\sigma^k](\Sigma^*)=\upc y_k$
for all $k\in\Nat$ (\Cref{def-pr}) and
$I_\sigma = \bigcap_{k\in\Nat}\Pre[\sigma^k](\Sigma^*) =\bigcap_k \upc y_k$.
From $y_0\subword y_1$ and monotonicity of $\pr$ (\Cref{lem-pre-mono})
we obtain $y_0\subword y_1\subword y_2\subword \cdots$ and $\upc
y_0\supseteq \upc y_1\supseteq \upc y_2\supseteq \cdots$. Thus we have
\[
I_\sigma
= \bigcap_{k\in\Nat} \upc y_k
= \begin{cases} \upc y_K  & \text{if  $y_{K}=y_{K+1}$ for some $K$,}
                        \\
                \emptyset & \text{if the $(y_k)_{k\in\Nat}$ sequence is strictly increasing.}
  \end{cases}
\]
\end{proof}

We write $\pr[\sigma^\omega](\epsilon)=y$ if $I_\sigma=\upc y$, and
$\pr[\sigma^\omega](\epsilon)=\bot$ if $I_\sigma$ is empty.

\begin{lemma}
\label{lem-sigma-omega-ptime}
$\pr[\sigma^\omega](\epsilon)$ can be computed in time $O(|\sigma|^3)$.
\end{lemma}
\begin{proof}[Proof (sketch)]
We start computing the elements $y_0,y_1,y_2,\ldots$ of the $(y_k)_k$
sequence.  If two consecutive values $y_{K}$ and $y_{K+1}$ coincide,
we have found $\pr[\sigma^\omega](\epsilon)$.  Otherwise we continue
while the sequence is strictly increasing until eventually
$|y_k|>|\rea(\sigma)|$ for some $k$ (indeed, some $k\leq 1+|\sigma|$).
In this case we can invoke \Cref{lem-up-sigma-um}.(ii) and conclude
that the $(y_k)_k$ sequence will remain strictly increasing, hence
$\pr[\sigma^\omega](\epsilon)=\bot$.

For complexity, we note that each $y_{k+1}$ is obtained in time
$O(|\sigma| + |y_k|)$ and has length in $O(|\sigma|^2)$ since
$|y_{k+1}|\leq |y_k| + |\rea(\sigma)|$ for all $k$.
\end{proof}

The set $I_\sigma$, represented via $\pr[\sigma^\omega](\epsilon)$,
is interesting because it characterises the
configurations from which a $\sigma$-labelled cycle can be traversed
infinitely many times, i.e., it characterises nontermination.

Indeed, the following lemma reduces nontermination to reachability:
\begin{lemma}[Existence of infinite runs]
\label{lem-infinite-runs}
(i) There exists an infinite sequence $x=x_0\step{\sigma}x_1
\step{\sigma}x_2 \cdots$ starting from $x$ if, and only if,
$\pr[\sigma^\omega](\epsilon)\subword x$.
\\
(ii) There exists an infinite run in $S$ that starts from $(q,x)$ and
visits a given $q'\in Q$ infinitely many times if, and only if, $q'$
is on an elementary cycle of $S$ and
$(q,x)\step{*}(q',\pr[\sigma_{q'}^\omega](\epsilon))$.
\end{lemma}
\begin{proof}
(i) Write $y$ for $\pr[\sigma^\omega](\epsilon)$. The proof of
\Cref{lem-sigma-omega-ptime} shows that, unless $y=\bot$,
$y=\pr[\sigma](y)$ and thus $y \step{\sigma} y$.  \\
($\DaG$): Since $x\supword y$, we have $x \step{\sigma} y
\step{\sigma} y \step{\sigma} \cdots$ if $\sigma \neq \epsilon$, and
$x \step{\sigma} x \step{\sigma} x \step{\sigma} \cdots$ in the
degenerate case where $\sigma = \epsilon$.
\\
($\GaD$): We assume $\sigma \neq \epsilon$ since otherwise $x \supword
\epsilon = \pr[\sigma^\omega](\epsilon)$ holds trivially.  The
infinite sequence $x_0 \step{\sigma} x_1 \step{\sigma} x_2
\step{\sigma} \cdots$ satisfies $x_0 \supword \pr[\sigma^k](x_k)
\supword \pr[\sigma^k](\epsilon)$ for all $k\in\Nat$. Thus
$\pr[\sigma^\omega](\epsilon)\neq\bot$
and $x=x_0
\supword  \pr[\sigma^\omega](\epsilon)$.
\\

\noindent
(ii) is an immediate consequence of (i).
\end{proof}

By combining the above lemmas with \Cref{thm-reach-in-NP} and the
$\NP$-hardness results proven in \Cref{app-NP-hardness}, one now obtains:
\begin{theorem}
\label{thm-nonterm-NPcomp}
Nontermination and existence of a B\"uchi run are $\NP$-complete for
flat LCMs.
\end{theorem}

\begin{remark}[Repeated coverability is $\NP$-complete]
Let us define more generally $I_\sigma(x)$ as
$\bigcap_{k=0,1,2,\ldots}\Pre[\sigma^k](\upc x)$, so that $I_\sigma$
really is shorthand for $I_\sigma(\epsilon)$.  For a location $q$ on a
$\sigma_q$-labelled cycle, $I_{\sigma_q}(x)$
characterises a form of \emph{repeated coverability}
since $y\in I_{\sigma_q}(x)$ iff there is an infinite run
from $(q,y)$ such that the channel contains a superword of $x$ every
time $q$ is (re)visited. Using some temporal logic, this could be
written under the form
\[
y\in I_{\sigma_q}(x) \iff (q,y) \models_\exists
\mathsf{G F}q\land\mathsf{G}(q\implies \textit{chan}\geq x)\:.
\]
The proof of \Cref{lem-sigma-omega-ptime} can be extended to the
computation of $I_\sigma(x)$. One obtains $I_{\sigma_q}(x)=\upc
y_0\cap \upc y_1\cap \cdots\cap\upc y_K$ for some $K$ in
$O(|\sigma|\cdot|x|)$.  We deduce that the  repeated coverability
problem is in $\NP$ for flat LCMs, and is indeed $\NP$-complete.

Note however that now the $(y_k)_k$ sequence does not necessarily
satisfies $y_0\subword y_1$, so that $I_\sigma(x)$ will have in
general several minimal elements, and possibly exponentially many.  In
fact already $\upc y_0\cap \upc y_1$ may  have exponentially many
minimal elements (see~\cite[\textsection~6.3]{ghkks-ideals}). Thus the
$\NP$-algorithm for repeated coverability  represents $I_\sigma(x)$ as
a conjunction of subword constraints, not via a set of minimal
elements, but this is sufficient for its purposes.
\qed
\end{remark}

Unboundedness reduces to reachability in a very similar way.
We say that a sequence of actions $\sigma$ is \emph{increasing} if
$u^{\ell_v}\subword v^{\ell_v-1}$ (and $\ell_v>0$) for $u\eqdef\rea(\sigma)$,
$v\eqdef\wri(\sigma)$ and $\ell_v\eqdef|v|$.
Now  $\pr[\sigma^\omega](\epsilon)$ and increasingness of $\sigma$
characterise
unbounded reachability sets.
\begin{lemma}[Proof in \Cref{app-unbounded-cns}]
\label{lem-unbounded-cns}
Let $x\in\Sigma^*$ be some channel contents and $\sigma$ a sequence of
channel actions. T.f.a.e.:\\
\noindent
\textit{(i)}
For all $k\in\Nat$ there exists $x_k$ with $x\step{\sigma^*}x_k$ and
$|x_k|\geq k$.
\\
\noindent
\textit{(ii)}
There exists an infinite unbounded sequence $x\step{\sigma^*}x_1
\step{\sigma^*}x_2 \step{\sigma^*}\cdots$ with
$|x_1|<|x_2|<\cdots$.
\\
\noindent
\textit{(iii)} $\sigma$ is increasing and $x \supword
\pr[\sigma^\omega](\epsilon)$.
\end{lemma}

\begin{lemma}[Existence of unbounded runs] In a flat LCM, t.f.a.e.\\
(i) The reachability set $\post^*(q,x)$ is infinite.
\\
(ii) There is an unbounded run starting from $(q,x)$.
\\
(iii)
$(q,x)\step{*}(q',\pr[\sigma_{q'}^\omega](\epsilon))$ for some control
location $q'$ with an increasing $\sigma_{q'}$.
\end{lemma}
\begin{proof}[Proof (sketch)]~\\
$(\textit{ii}\implies\textit{iii})$:
In an unbounded run, there must be a control location $q'$ that is visited infinitely
many times with  associated channel contents that are unbounded. Since
 from $q'$ one can only return to $q'$ by running through the cycle
 around $q'$, hence performing  $\sigma_{q'}$ some number of times,
 the first visit of $q'$ is some $(q',x')$ satisfying
 case \textit{(ii)}  of \Cref{lem-unbounded-cns}. We deduce that
 $\sigma'_q$ is increasing and that
 $x'\supword\pr[\sigma_{q'}^\omega](\epsilon)$ as in case
 \textit{(iii)}  of the Lemma.
\\
\noindent
$(\textit{iii}\implies\textit{ii})$: by \Cref{lem-unbounded-cns} there
exists an unbounded run starting from
$(q',\pr[\sigma_{q'}^\omega](\epsilon))$. Hence there is one starting
from $(q,x)$.  \\
\noindent
$(\textit{i}\iff\textit{ii})$: is an application of K\H{o}nig's Lemma,
not specific to LCMs, see e.g.~\cite[\textsection 6]{phs-rp2010}.
\end{proof}
We can thus reduce unboundedness to reachability of an increasing
cycle.  With the $\NP$-hardness results proven in
\Cref{app-NP-hardness}, one now obtains:
\begin{theorem}
\label{thm-unbound-NPcomp}
Unboundedness for flat LCMs is $\NP$-complete.
\end{theorem}

\section{Conclusion}
We analysed the behaviour of the backward-reachability algorithm for
lossy channel machines when a cycle of channel actions can be
performed arbitrarily many times. This provides complexity bounds on
the size of runs that follow a bounded path scheme of the form
$\sigma_1^*\rho_1\sigma_2^*\rho_2\ldots\sigma_m^*\rho_m$, with
applications in the verification of flat systems, or in bounded
verification for general systems. The main result is an $\NP$ upper
bound for reachability and, by reduction, several other verification
problems like unboundedness or existence of a B\"uchi run.

Natural directions for future work include extending our approach to
deal with richer verification problems, like temporal logic model
checking.  It would also be interesting to consider more expressive
models, like the partially lossy channel systems
from~\cite{kocher2019} or the higher-order lossy channel systems and
priority channel systems from~\cite{HSS-lmcs}.

\section*{Acknowledgements}
We thank A.\ Finkel who raised the issue of flatness in lossy
channel systems. We also thank J.\ Leoux and S.\ Halfon for
useful comments that helped improve this paper.

\bibliography{shorted}

\appendix

\section{Some SLP algorithms}
\label{sec-new-SLP-algos}

We describe here some SLP algorithms that are not readily available in
the literature (as far as we know).  Formally, by ``an SLP $X$'' we
mean a grammar $(\Sigma,N,X,P)$ where $X\in N$ is the axiom (a non
terminal), where $\Sigma$ is the set of terminal letters, and where
the production rules in $P$ are either $A_i\to a$ or $A_i\to A_j A_k$
for some $a\in \Sigma$ and some nonterminals $A_i,A_j,A_k$ with
$i<j,k$. There is exactly one production rule for each $A_i\in N$, so
that each $A_i$ defines a unique word $L(A_i)\in\Sigma^*$.

\subsection{Deciding $X\subword v^n$}
\label{app-SLP-subword}

Deciding $X\subword Y$ between SLPs is a difficult
problem, \PP-hard as show in~\cite{lohrey2012}.
When $x$ (or $y$) is a plain word, the problem has polynomial-time solutions~\cite{MarSch-IPL2004,cegielski2006,yamamoto2011}.

Here we consider the special case where $Y$ is some $v^n$.

\begin{proposition}
\label{prop-vn-subword-X}
Deciding whether $X\subword v^n$, where $X$ is an SLP, $v$ is a plain
word, and $n$ is a fractional exponent, can be done in time
$O(\size{X}\cdot|v|+|v|^2+\log n)$.
\end{proposition}
\begin{proof}
For $v\neq\epsilon$ and some word $x$ such that
$\alphabet(x)\subseteq\alphabet(v)$, let us define $p(x,v)$ as the smallest
fractional power such that $x\subword v^p$.  Now $p(x,v)$ satisfies
the following equalities:
\begin{equation}
\label{eq-pvx}
\begin{aligned}
p(\epsilon,v)&=0
\\
p(a,v)&=
        \frac{i}{|v|}, \text{ if the first occurrence of $a$ in $v$ is
        at position $i$,}
\\
p(x\:y,v)&=p(x,v)+p(y,v_{(j)}), \text{ if $p(x,v)$ is some
$q+\frac{j}{|v|}$ with $q\in\Nat$.}
\end{aligned}
\end{equation}
Using \cref{eq-pvx} leads to a dynamic programming algorithm computing
$p(X,v)$ for an SLP $X$. After checking that
$\alphabet(X)\subseteq\alphabet(v)$, one computes the values of all
$p(A,v_{(i)})$ for $i=1,\ldots,|v|$ and $A$ a nonterminal in SLP
$X$. Each of these $O(\size{X}\cdot|u|)$ values is computed in time
$O(1)$ if one precomputes the first occurrences of letters in the
cyclic shifts of $v$, say in time $O(|v|^2)$.  Finally, one only has
to compare $p(X,v)$ with $n$.
\end{proof}

\subsection{Computing $X/v^k$}
\label{ssec-residual-via-SLP}

\begin{proposition}
\label{prop-X-remove-vn}
Building a SLP for $X/v^k$, where $X$ is an SLP, $v$ is a plain word,
and $k\in\Nat$, can be done in time $\poly(\size{X}+|v|+\log n)$.
\end{proposition}
\begin{proof}
For given $\ell$, deciding whether $X/v^k$ has length at least $\ell$
is easy: One just applies the definition, builds an SLP $X'$ for the
suffix of length $|X|-\ell$ of $X$, and checks that it is a subword of
$v^k$ with \cref{prop-vn-subword-X}.

Thus one can computes $|X/v^k|$ by finding the length of the result
via dichotomic search, repeating the previous process $\log |X|$,
i.e., $O(\size{X})$, times.\footnote{A better, dynamic
programming, algorithm exists but here we aim for the simplest feasability proof.}
\end{proof}

\section{Forward reachability techniques}
\label{app-forward}

We collect in this section some proofs relying on forward-reachability
analysis.

Let us reuse notations from~\cite{abdulla-forward-lcs} and define a
partial function $x\ominus u$ between channel contents as follows:
\begin{equation}
x \ominus u \eqdef
\begin{cases}
        x           &\text{ if $u=\epsilon$,} \\
   \text{undefined} &\text{ if $x=\epsilon$ and $u\neq\epsilon$,} \\
   x'\ominus u' &\text{ if $x=a x'$ and $u=a u'$ for some $a\in\Sigma$,} \\
   x'\ominus u &\text{ if $x=a x'$ and $u=bu'$ for some $a\neq b\in\Sigma$.}
\end{cases}
\end{equation}
Observe that $x\ominus u$ is defined if, and only if, $u\subword x$.
Note also that, when $x\ominus u$ is defined, we can use monotonicity
and commutation with concatenation:
\begin{equation}
\label{eq-mono-and-comm}
\text{if } u\subword x \text{ then for all }
x':
\begin{cases}
x\subword x' \text{ implies } x\ominus u\subword x'\ominus u    \:,
\\
(x\ominus u)\cdot x' = (xx')\ominus u   \:.
\end{cases}
\end{equation}

\medskip

\noindent
Now $x\ominus u$ captures the forward effects of $?u$ actions in LCMs:
\begin{lemma}
$x\step{?u}y$ iff $y\subword x\ominus u$.
\end{lemma}
We can also use $\ominus$ to characterise the outcome of arbitrary
sequences of actions.
\begin{lemma}
\label{lem-step-vs-ominus}
Let  $\sigma\in\Act_\Sigma^*$ be an arbitrary sequence
of actions.
\[
x'\step{\sigma}y
\text{ for some } x'\subword x
\text{ iff }
x\step{\sigma} \text{ and } y\subword
\bigl(x\cdot\wri(\sigma)\bigr)\ominus\rea(\sigma) \:.
\]
\end{lemma}
\begin{proof}
By induction on the length of $\sigma$. The existential quantification
on some $x'\subword x$ accounts for the case where $\sigma=\epsilon$
is the empty sequence.

For the inductive step, we consider two cases:
\begin{enumerate}
\item
$\sigma=\:!w\cdot\sigma'$:
For the ``$\GaD$'' direction,
  $x\step{\sigma}y$ implies $x'\step{\sigma'}y$ for some $x'\subword x
  w$, which implies
\begin{xalignat*}{2}
y &\subword \bigl(x' w\cdot\wri(\sigma')\bigr) \ominus \rea(\sigma')
&&\text{by ind.\ hyp.,}
\\
 &= \bigl(x'\cdot\wri(\sigma)\bigr) \ominus \rea(\sigma)
&&\text{since $\sigma=\:!w\cdot\sigma'$,}
\\
 &\subword \bigl(x\cdot\wri(\sigma)\bigr) \ominus \rea(\sigma)
&&\text{by monotonicity.}
\end{xalignat*}
For the ``$\DaG$'' direction, we know that
$y \subword \bigl(x.\wri(\sigma)\bigr)\ominus\rea(\sigma)
 = \bigl(x w.\wri(\sigma')\bigr)\ominus\rea(\sigma')$, so the
ind.\ hyp.\ tells us that $x'\step{\sigma'}y$ for some $x'\subword
xw$. We deduce $x\step{!w}x'\step{\sigma'}y$.

\item
$\sigma=\:?w\cdot\sigma'$:
For the ``$\GaD$'' direction, $x'\step{\sigma}y$ implies $x'\step{?w}x''\step{\sigma'}y$ for
some $x''\subword x'\ominus w$.
We have
\begin{xalignat*}{2}
y&\subword \bigl(x''.\wri(\sigma')\bigr)\ominus\rea(\sigma')
&&\text{by ind.\ hyp.,}
\\
&\subword \bigl([x\ominus w]\cdot\wri(\sigma')\bigr)\ominus\rea(\sigma')
&&\text{by monotonicity,}
\\
&=\bigl(x\cdot\wri(\sigma')\bigr)\ominus \bigl(w\cdot\rea(\sigma')\bigr)
&&\text{by \eqref{eq-mono-and-comm},}
\\
&=(x\cdot\wri(\sigma))\ominus \rea(\sigma)
&&\text{since $\sigma = \:?w\cdot\sigma'$.}
\end{xalignat*}
For the ``$\DaG$'' direction, we know that
$x\step{\sigma}$ hence in particular $x\ominus w$ is defined.
We also know that
$y\subword\bigl(x.\wri(\sigma)\bigr)\ominus\rea(\sigma)
=\bigl(x.\wri(\sigma')\bigr)\ominus(w\cdot\rea(\sigma'))
=\bigl((x\ominus w).\wri(\sigma')\bigr)\ominus\rea(\sigma')$, so by
ind.\ hyp.\ there is some $x'\subword x\ominus w$ with
 $x'\step{\sigma'}y$ for some $x'\subword x w$. We deduce
$x\step{?w}x'\step{\sigma'}y$.
\end{enumerate}
\end{proof}

\subsection{Proof of \Cref{lem-unbounded-cns}}
\label{app-unbounded-cns}

Write $u$, $v$ for $\rea(\sigma)$, $\wri(\sigma)$.

\medskip

\noindent
$(\textit{ii}\implies\textit{iii})$: we only have to prove that
$\sigma_q$ is increasing since \Cref{lem-infinite-runs} entails
$x\supword \pr[\sigma^\omega](\epsilon)$ already.

By assumption, there is a sequence $x_1,x_2,\ldots$ of channel
contents of increasing length, and some numbers
$n_1,n_2,\ldots$ in $\Nat$ such that
$x\step{\sigma^{n_i}}x_{i}$. W.l.o.g.\ we
can assume $n_1<n_2<\cdots$.

With \Cref{lem-step-vs-ominus} we deduce $x_i\subword(x \:
v^{n_i})\ominus u^{n_i}$, hence $u^{n_i} \: x_i\subword x \: v^{n_i}$,
for all $i=1,2,\ldots$
If $u=\epsilon$, $\sigma$ is trivially increasing, so assume $|u|>0$
and write $m=|x|$:
we get $u^{n_i-m}x_i \subword v^{n_i}$ for all $i$ such that $n_i\geq m$.
Now take $i$ such that
$|x_i|\geq(m+1)|v|$ (and such that $n_i>m$): we get
$u^{n_i-m} \subword v^{n_i-m-1}$.
We now applies Lemma~6.2
from~\cite{abdulla-forward-lcs}: ``if there is some $k\geq 1$ such
that $w_1^k\subword w_2^{k-1}$ (for two words $w_1,w_2$), then in
particular one can choose $k=|w_2|$''. This yields $u^{|v|}\subword
v^{|v|-1}$, i.e., $\sigma$ is increasing.
\\

\noindent
$(\textit{iii}\implies\textit{i})$: we assume that $\sigma$ is
increasing, i.e., $u^{|v|}\subword v^{|v|-1}$, and that $x\supword
\pr[\sigma_q^\omega](\epsilon)$.  The second assumption entails that
$x\step{\sigma^n}$ for all $n$.  The first assumption entails
$v^k \subword x \: v^{k|v|}\ominus u^{k|v|}$, hence
$x\step{\sigma^{k|v|}}v^k$ by \Cref{lem-step-vs-ominus}, for all $k\in\Nat$.\\

\noindent
$(\textit{i}\implies\textit{ii})$:
is an application of K\H{o}nig's Lemma, not specific to LCMs, see
e.g.~\cite[\textsection 6]{phs-rp2010}.

\section{$\NP$-hardness for flat LCMs and flat FIFO machines}
\label{app-NP-hardness}

LCMs are derived from FIFO automata~\cite{vauquelin80,brand83} and
our $\NP$-hardness results apply to both models. \emph{FIFO automata}, sometimes
called \emph{queue automata}, or \emph{communicating finite state machines}, are
reliable channel machines where messages are never lost. Their
operational semantics is based on a reliable notion of steps, formally
given by $x\relstep{!\:w}y \;\equivdef\; y= x w$ and $x\relstep{?\:w}y
\;\equivdef\; w y= x$, to be compared with
\Cref{eq-def-sem-actions}. This is extended to $x\relstep{\sigma}y$,
$c\relstep{*}c'$, etc., as for LCMs.

\subsection{Proof of \Cref{thm-hard-acyclic}: $\NP$-hardness for acyclic machines}
\label{app-nphard-cycle-free}

We first show hardness for reachability and reduce from $\SAT$. Let
$\phi=C_1\land\cdots\land C_m$ be a 3CNF with Boolean variables among
$V=\{v_1,\ldots,v_n\}$. With $\phi$ we associate a machine $S_\phi$ as
illustrated below in \cref{fig-reduction-SAT2LCM}.

Let us explain informally how $S_\phi$ operates. Starting from
$I^{\text{b}}$ it first reaches $I^{\text{e}}$ while writing in the
channel a word of the form $w\:\ttdol$ with $w\in\{0,1\}^n$.  This
word encodes a valuation of the Boolean variables and carries an end
marker $\ttdol$.  Then $S_\phi$ crosses from $C_1^{\text{b}}$ to
$C_1^{\text{e}}$: this requires reading the valuation on the channel
and checking that it satisfies $C_1$. For this $S_\phi$ has to choose
the line corresponding to one of the three literals in $C_1$, in fact
choose one literal made true by the valuation. During this check, the
valuation is written back on the channel.  Then $S_\phi$ checks that
the remaining clauses, $C_2$ to $C_m$, are satisfied by the valuation,
each time reading the valuation and writing it back on the channel.
Finally, the last leg from $V^{\text{b}}$ to $V^{\text{e}}$ checks
that no message has been lost during all this run.
\begin{figure}[htbp]
\centering
\scalebox{0.77}{
  \begin{tikzpicture}[->,>=stealth',shorten >=1pt,node distance=6em,thick,auto,bend angle=30]
\tikzstyle{every state}=[minimum size=2em,initial where=above,initial text={}]

\node at (0,10)	   [state,initial] (c0) {\footnotesize $I^{\text{b}}$};
\node at (1.5,10)  [state]	   (c1) {\footnotesize $1$};
\node at (3,10)	   [state]	   (c2) {\footnotesize $2$};
\node at (4.5,10)  [state]	   (c3) {\footnotesize $3$};
\node at (6,10)	   [state]	   (c4) {\footnotesize $4$};
\node at (7.5,10)  [state,dashed]  (c5) {};
\node at (9,10)	   [state,dashed]  (c6) {};
\node at (10.5,10) [state]	   (c7) {$n$};
\node at (12,10)   [state]	   (c8) {\footnotesize $I^{\text{e}}$};

\node at ($(c5)!0.5!(c6)$) [minimum size=0em,inner sep=0.4em] (cdots) {\Large $\cdots$};

\path (c0) edge [bend left=20]	       node	   {$!0$}  (c1);
\path (c0) edge [bend right=20]	       node [swap] {$!1$}  (c1);
\path (c1) edge [bend left=20]	       node	   {$!0$}  (c2);
\path (c1) edge [bend right=20]	       node [swap] {$!1$}  (c2);
\path (c2) edge [bend left=20]	       node	   {$!0$}  (c3);
\path (c2) edge [bend right=20]	       node [swap] {$!1$}  (c3);
\path (c3) edge [bend left=20]	       node	   {$!0$}  (c4);
\path (c3) edge [bend right=20]	       node [swap] {$!1$}  (c4);
\path (c4) edge [bend left=20]	       node	   {$!0$}  (c5);
\path (c4) edge [bend right=20]	       node [swap] {$!1$}  (c5);
\path (c6) edge [bend left=20,dashed]  node	   {$!0$}  (c7);
\path (c6) edge [bend right=20,dashed] node [swap] {$!1$}  (c7);
\path (c7) edge			       node	   {$!\ttdol$} (c8);

\node at (0,8)	  [state]	 (d0) {\footnotesize $C^{\text{b}}_1$};
\node at (1.5,8)  [state]	 (d1) {};
\node at (3,8)	  [state]	 (d2) {};
\node at (4.5,8)  [state]	 (d3) {};
\node at (6,8)	  [state]	 (d4) {};
\node at (7.5,8)  [state,dashed] (d5) {};
\node at (9,8)	  [state,dashed] (d6) {};
\node at (10.5,8) [state]	 (d7) {};
\node at (12,8)	  [state]	 (d8) {};

\node at ($(d5)!0.5!(d6)$) [minimum size=0em,inner sep=0.4em] (ddots) {\Large $\cdots$};
\path (d0) edge [color=red]	       node	   {$?1\:!1$}	(d1);
\path (d1) edge [bend left=20]	       node	   {$?0\:!0$}	(d2);
\path (d1) edge [bend right=20]	       node [swap] {$?1\:!1$}	(d2);
\path (d2) edge [bend left=20]	       node	   {$?0\:!0$}	(d3);
\path (d2) edge [bend right=20]	       node [swap] {$?1\:!1$}	(d3);
\path (d3) edge [bend left=20]	       node	   {$?0\:!0$}	(d4);
\path (d3) edge [bend right=20]	       node [swap] {$?1\:!1$}	(d4);
\path (d4) edge [bend left=20]	       node	   {$?0\:!0$}	(d5);
\path (d4) edge [bend right=20]	       node [swap] {$?1\:!1$}	(d5);
\path (d6) edge [bend left=20,dashed]  node	   {$?0\:!0$}	(d7);
\path (d6) edge [bend right=20,dashed] node [swap] {$?1\:!1$}	(d7);
\path (d7) edge			       node	   {$?\ttdol\:!\ttdol$} (d8);

\node at (0,6)	  [state]	 (x0) {};
\node at (1.5,6)  [state]	 (x1) {};
\node at (3,6)	  [state]	 (x2) {};
\node at (4.5,6)  [state]	 (x3) {};
\node at (6,6)	  [state]	 (x4) {};
\node at (7.5,6)  [state,dashed] (x5) {};
\node at (9,6)	  [state,dashed] (x6) {};
\node at (10.5,6) [state]	 (x7) {};
\node at (12,6)	  [state]	 (x8) {};

\node at ($(x5)!0.5!(x6)$) [minimum size=0em,inner sep=0.4em] (xdots) {\Large $\cdots$};

\path (x0) edge [bend left=20]	       node	   {$?0\:!0$}	(x1);
\path (x0) edge [bend right=20]	       node [swap] {$?1\:!1$}	(x1);
\path (x1) edge [color=red]	       node	   {$?0\:!0$}	(x2);
\path (x2) edge [bend left=20]	       node	   {$?0\:!0$}	(x3);
\path (x2) edge [bend right=20]	       node [swap] {$?1\:!1$}	(x3);
\path (x3) edge [bend left=20]	       node	   {$?0\:!0$}	(x4);
\path (x3) edge [bend right=20]	       node [swap] {$?1\:!1$}	(x4);
\path (x4) edge [bend left=20]	       node	   {$?0\:!0$}	(x5);
\path (x4) edge [bend right=20]	       node [swap] {$?1\:!1$}	(x5);
\path (x6) edge [bend left=20,dashed]  node	   {$?0\:!0$}	(x7);
\path (x6) edge [bend right=20,dashed] node [swap] {$?1\:!1$}	(x7);
\path (x7) edge			       node	   {$?\ttdol\:!\ttdol$} (x8);

\node at (0,4)	  [state]	 (y0) {};
\node at (1.5,4)  [state]	 (y1) {};
\node at (3,4)	  [state]	 (y2) {};
\node at (4.5,4)  [state]	 (y3) {};
\node at (6,4)	  [state]	 (y4) {};
\node at (7.5,4)  [state,dashed] (y5) {};
\node at (9,4)	  [state,dashed] (y6) {};
\node at (10.5,4) [state]	 (y7) {};
\node at (12,4)	  [state]	 (y8) {\footnotesize $C^{\text{e}}_1$};

\node at ($(y5)!0.5!(y6)$) [minimum size=0em,inner sep=0.4em] (ydots) {\Large $\cdots$};

\path (y0) edge [bend left=20]	       node	   {$?0\:!0$}	(y1);
\path (y0) edge [bend right=20]	       node [swap] {$?1\:!1$}	(y1);
\path (y1) edge [bend left=20]	       node	   {$?0\:!0$}	(y2);
\path (y1) edge [bend right=20]	       node [swap] {$?1\:!1$}	(y2);
\path (y2) edge [bend left=20]	       node	   {$?0\:!0$}	(y3);
\path (y2) edge [bend right=20]	       node [swap] {$?1\:!1$}	(y3);
\path (y3) edge [color=red]	       node	   {$?1\:!1$}	(y4);
\path (y4) edge [bend left=20]	       node	   {$?0\:!0$}	(y5);
\path (y4) edge [bend right=20]	       node [swap] {$?1\:!1$}	(y5);
\path (y6) edge [bend left=20,dashed]  node	   {$?0\:!0$}	(y7);
\path (y6) edge [bend right=20,dashed] node [swap] {$?1\:!1$}	(y7);
\path (y7) edge			       node	   {$?\ttdol\:!\ttdol$} (y8);

\path (d0) edge (x0);
\path (x0) edge (y0);
\path (d8) edge (x8);
\path (x8) edge (y8);

\node at (0,2)	  [state,dashed] (m0) {\footnotesize $C^{\text{b}}_2$};
\node at (1.5,2)  []		 (m1) {};
\node at (12,1.5) []		 (m8) {};
\node at (9.5,0)  []		 (n6) {};
\node at (0,0.5)  []		 (n0) {};
\node at (10.5,0) [state,dashed] (n7) {};
\node at (12,0)	  [state]	 (n8) {\footnotesize $C^{\text{e}}_m$};
\path (m0) edge [dashed]		   (m1);
\path (m0) edge [dashed]		   (n0);
\path (n6) edge [dashed]		   (n7);
\path (m8) edge [dashed]		   (n8);
\path (n7) edge	 node  {$?\ttdol\:!\ttdol$} (n8);

\node at ($(m0)!0.5!(n8)$) [minimum size=0em,inner sep=0.4em] (edots) {\Large $\cdots$};

\node at (0,-2)	   [state]	  (e0) {\footnotesize $V^{\text{b}}$};
\node at (1.5,-2)  [state]	  (e1) {};
\node at (3,-2)	   [state]	  (e2) {};
\node at (4.5,-2)  [state]	  (e3) {};
\node at (6,-2)	   [state]	  (e4) {};
\node at (7.5,-2)  [state,dashed] (e5) {};
\node at (9,-2)	   [state,dashed] (e6) {};
\node at (10.5,-2) [state]	  (e7) {};
\node at (12,-2)   [state]	  (e8) {\footnotesize $V^{\text{e}}$};

\node at ($(e5)!0.5!(e6)$) [minimum size=0em,inner sep=0.4em] (edots) {\Large $\cdots$};

\path (e0) edge [bend left=20]	       node	   {$?0$}  (e1);
\path (e0) edge [bend right=20]	       node [swap] {$?1$}  (e1);
\path (e1) edge [bend left=20]	       node	   {$?0$}  (e2);
\path (e1) edge [bend right=20]	       node [swap] {$?1$}  (e2);
\path (e2) edge [bend left=20]	       node	   {$?0$}  (e3);
\path (e2) edge [bend right=20]	       node [swap] {$?1$}  (e3);
\path (e3) edge [bend left=20]	       node	   {$?0$}  (e4);
\path (e3) edge [bend right=20]	       node [swap] {$?1$}  (e4);
\path (e4) edge [bend left=20]	       node	   {$?0$}  (e5);
\path (e4) edge [bend right=20]	       node [swap] {$?1$}  (e5);
\path (e6) edge [bend left=20,dashed]  node	   {$?0$}  (e7);
\path (e6) edge [bend right=20,dashed] node [swap] {$?1$}  (e7);
\path (e7) edge			       node	   {$?\ttdol$} (e8);

\node at (12.5,10) [label=right:{Write some valuation $v$}]				    (lc) {} ;
\node at (12.5,6)  [label=right:{Check $v\models C_1$ ($\equiv\:v_1\lor\neg v_2\lor v_4$)}] (lx) {} ;
\node at (12.5,1)  [label=right:{Check $v\models C_2\land \cdots\land C_m$}]		    (lm) {};
\node  at (12.5,-2)  [label=right:{Check no losses occurred}] {};

\draw (c8) .. controls (d8) and (c0) .. (d0);
\draw (y8) .. controls (m8) and (y0) .. (m0);
\draw (n8) .. controls (e8) and (n0) .. (e0);

  \end{tikzpicture}
}
\caption{LCM $S_\phi$ for satisfiability of $\phi = (v_1\lor \neg v_2\lor v_4)\land C_2\cdots \land C_m$.}
\label{fig-reduction-SAT2LCM}
\end{figure}
 
It is now clear that $(I^{\text{b}},\epsilon) \step{*}
(V^{\text{e}},\epsilon)$ in $S_\phi$ if, and only if, $\phi$ is
satisfiable. The reasoning holds for lossy LCMs and for reliable FIFO
automata.  We have thus reduced $\SAT$ to the reachability problem for
both types of acyclic machines.

\begin{remark}
The construction of $S_\phi$ can be simplified at the cost of making
the reduction perhaps less obviously correct: one can either omit the
end-marker symbol $\ttdol$ since in the end the machine checks that no
message was lost (thus a binary alphabet suffices), or one can stop
the machine at $C_m^{\text{e}}$, getting rid of the $V^{\text{b}}$ to
$V{\text{e}}$ part, since the markers ensure that the valuation read
while checking a clause $C_i$ is indeed the full valuation written at
the previous stage.
\qed
\end{remark}

For hardness of nontermination and unboundedness
we adapt the previous reduction by adding a single cycle
$V^{\text{e}}\step{!\ttdol}V^{\text{e}}$ on the last
control location. Starting from $(I^{\text{b}},\epsilon)$, the modified
$S_\phi$ has an infinite run iff it has an unbounded run iff $\phi$ is
satisfiable.
\\

The above reductions adapt to flat VASSes and lossy VASSes, i.e.,
channel machines with unary alphabet, provided that we allow $2n$
channels (or counters) for a valuation on $n$ Boolean variables.

\subsection{Proof of \Cref{thm-hard-single-path}: $\NP$-hardness for single-path machines}
\label{app-nphard-single-path}

We first show hardness for reachability.  For this we reduce from
$\SAT$.  So let us consider a 3CNF formula $\phi$ with Boolean
variables among $V=\{v_1,\ldots,v_n\}$. Let us say $\phi=(v_2\lor\neg
v_3\lor\neg v_n)\land C_2\land \cdots\land C_m$, with $m$ clauses.

With $\phi$ we associate $S_\phi$, the single-path flat LCM described in
\Cref{fig-reduction-singlepath}.
This LCM has $O(m n^2)$ control locations\footnote{Our reduction
insists on using only one channel. With multiple channels the same
idea would use $O(n+m)$ control locations.}, and is organised as a
series of distinct operations on the channel contents.
\begin{figure}[htbp]
\centering
$\!\!\!\!\!\!\!\!\!\!\!$\scalebox{0.75}{
  \begin{tikzpicture}[->,>=stealth',shorten >=1pt,node distance=6em,thick,auto,bend angle=30]
\tikzstyle{every state}=[minimum size=1.3em,inner sep=0em]

\node at (2,10)	    [state,initial] (c0)      {\tiny $0$};
\node at (5.2,10)   [state]	    (c1a)     {\tiny $0,1$};
\node at (7.15,10)  [state]	    (c1b)     {};
\node at (8.15,10)  [state]	    (c1c)     {};
\node at (9.15,10)  [state]	    (c2a)     {\tiny $0,2$};
\node at (11.1,10)  [state]	    (c2b)     {};
\node at (12.1,10)  [state]	    (c2c)     {};
\node at (13.1,10)		    (c3ghost) {};
\node at (14.05,10)		    (cnghost) {};
\node at (15.05,10) [state]	    (cna)     {\tiny $0,n$};
\node at (17,10)    [state]	    (cnb)     {};
\node at (18,10)    [state]	    (cnc)     {};

\node at ($(c3ghost)!0.5!(cnghost)$) [minimum size=0em,inner sep=0.4em] (Cdots) {\Large $\cdots$};

\path (c0)	edge	      node  {$!\ttv_10\ttv_20\cdots\ttv_n0$} (c1a);
\path (c1a)	edge	      node  {$?\ttv_1 \: !\ttv_1$}	     (c1b);
\path (c1b)	edge	      node  {}				     (c1c);
\path (c1c)	edge	      node  {}				     (c2a);
\path (c2a)	edge	      node  {$?\ttv_2 \: !\ttv_2$}	     (c2b);
\path (c2b)	edge	      node  {}				     (c2c);
\path (c2c)	edge [dashed] node  {}				     (c3ghost);
\path (cnghost) edge [dashed] node  {}				     (cna);
\path (cna)	edge	      node  {$?\ttv_n \: !\ttv_n$}	     (cnb);
\path (cnb)	edge	      node  {}				     (cnc);

\path (c1b) edge [loop above]		node  {$?0 \: !0$} (c1b);
\path (c1c) edge [loop above,color=red] node  {$?0 \: !1$} (c1c);
\path (c2b) edge [loop above]		node  {$?0 \: !0$} (c2b);
\path (c2c) edge [loop above,color=red] node  {$?0 \: !1$} (c2c);
\path (cnb) edge [loop above]		node  {$?0 \: !0$} (cnb);
\path (cnc) edge [loop above,color=red] node  {$?0 \: !1$} (cnc);

\node at (0,8)	  [state]	 (d11) {\tiny $1,1$};
\node at (1.8,8)  [state]	 (d12) {};
\node at (2.8,8)  [state]	 (d13) {};
\node at (3.8,8)  [state]	 (d21) {\tiny $1,2$};
\node at (5.6,8)  [state]	 (d22) {};
\node at (6.6,8)  [state]	 (d23) {};
\node at (7.6,8)  [state]	 (d31) {\tiny $1,3$};
\node at (9.4,8)  [state]	 (d32) {};
\node at (10.4,8) [state]	 (d33) {};
\node at (11.4,8) [state]	 (d41) {\tiny $1,4$};
\node at (13.2,8) [state,dashed] (d42) {};
\node at (15.2,8) [state]	 (dn1) {\tiny $1,n$};
\node at (17,8)	  [state]	 (dn2) {};
\node at (18,8)	  [state]	 (dn3) {};

\node at ($(d42)!0.5!(dn1)$) [minimum size=0em,inner sep=0.4em] (Ddots) {\Large $\cdots$};

\path (d11) edge  node	{$?\ttv_1 \: !\ttv_1$} (d12);
\path (d12) edge  node	{}		       (d13);
\path (d13) edge  node	{}		       (d21);
\path (d21) edge  node	{$?\ttv_2 \: !\ttv_2$} (d22);
\path (d22) edge  node	{}		       (d23);
\path (d23) edge  node	{}		       (d31);
\path (d31) edge  node	{$?\ttv_3 \: !\ttv_3$} (d32);
\path (d32) edge  node	{}		       (d33);
\path (d33) edge  node	{}		       (d41);
\path (d41) edge  node	{$?\ttv_4 \: !\ttv_4$} (d42);
\path (dn1) edge  node	{$?\ttv_n \: !\ttv_n$} (dn2);
\path (dn2) edge  node	{}		       (dn3);

\path (d12) edge [loop above]		node			  {$?0 \: !0$}	   (d12);
\path (d13) edge [loop above]		node			  {$?1 \: !1$}	   (d13);
\path (d22) edge [loop above]		node			  {$?0 \: !0$}	   (d22);
\path (d23) edge [loop above,color=red] node [shift={(1ex,0ex)}]  {$?1 \: !1\ttx$} (d23);
\path (d32) edge [loop above,color=red] node [shift={(-1ex,0ex)}] {$?0 \: !0\ttx$} (d32);
\path (d33) edge [loop above]		node			  {$?1 \: !1$}	   (d33);
\path (dn2) edge [loop above,color=red] node [shift={(-1ex,0ex)}] {$?0 \: !0\ttx$} (dn2);
\path (dn3) edge [loop above]		node			  {$?1 \: !1$}	   (dn3);

\node at (0,6)	  [state]	       (e0)	  {};
\node at (0.9,6)  [state]	       (e11)	  {\tiny $2,1$};
\node at (1.8,6)  [state]	       (e12)	  {};
\node at (2.7,6)  [state]	       (e13)	  {};
\node at (3.6,6)  [state]	       (e14)	  {};
\node at (4.5,6)  [state]	       (e15)	  {};
\node at (5.4,6)  [state]	       (e16)	  {};
\node at (6.9,6)  [state]	       (e21)	  {\tiny $2,2$};
\node at (7.8,6)  [state]	       (e22)	  {};
\node at (8.7,6)  [state]	       (e23)	  {};
\node at (9.6,6)  [state]	       (e24)	  {};
\node at (10.5,6) [state]	       (e25)	  {};
\node at (11.4,6) [state]	       (e26)	  {};
\node at (12.3,6) [minimum size=0.1em] (e27ghost) {};
\node at (13.5,6) [state]	       (en1)	  {\tiny $2,n$};
\node at (14.4,6) [state]	       (en2)	  {};
\node at (15.3,6) [state]	       (en3)	  {};
\node at (16.2,6) [state]	       (en4)	  {};
\node at (17.1,6) [state]	       (en5)	  {};
\node at (18,6)	  [state]	       (en6)	  {};

\node at ($(e27ghost)!0.35!(en1)$) [minimum size=0em,inner sep=0.4em] (Edots) {\Large $\cdots$};

\path (e0) edge (e11);
\path (e11) edge (e12);
\path (e12) edge (e13);
\path (e13) edge (e14);
\path (e14) edge (e15);
\path (e15) edge (e16);
\path (e16) edge (e21);
\path (e21) edge (e22);
\path (e22) edge (e23);
\path (e23) edge (e24);
\path (e24) edge (e25);
\path (e25) edge (e26);
\path (e26) edge [dashed] (e27ghost);
\path (en1) edge (en2);
\path (en2) edge (en3);
\path (en3) edge (en4);
\path (en4) edge (en5);
\path (en5) edge (en6);

\path (e0)  edge [loop above,color=red] node  {$?\ttx \: !\ttx$}	     (e0);
\path (e11) edge [loop below]		node  {$?\ttv_1 \: !\ttv_1$}	     (e11);
\path (e12) edge [loop above,color=red] node  {$?\ttv_1\ttx \: !\ttx\ttv_1$} (e12);
\path (e13) edge [loop below]		node  {$?0 \: !0$}		     (e13);
\path (e14) edge [loop above,color=red] node  {$?0\ttx \: !\ttx0$}	     (e14);
\path (e15) edge [loop below]		node  {$?1 \: !1$}		     (e15);
\path (e16) edge [loop above,color=red] node  {$?1\ttx \: !\ttx 1$}	     (e16);
\path (e21) edge [loop below]		node  {$?\ttv_2 \: !\ttv_2$}	     (e21);
\path (e22) edge [loop above,color=red] node  {$?\ttv_2\ttx \: !\ttx\ttv_2$} (e22);
\path (e23) edge [loop below]		node  {$?0 \: !0$}		     (e23);
\path (e24) edge [loop above,color=red] node  {$?0\ttx \: !\ttx0$}	     (e24);
\path (e25) edge [loop below]		node  {$?1 \: !1$}		     (e25);
\path (e26) edge [loop above,color=red] node  {$?1\ttx \: !\ttx 1$}	     (e26);
\path (en1) edge [loop below]		node  {$?\ttv_n \: !\ttv_n$}	     (en1);
\path (en2) edge [loop above,color=red] node  {$?\ttv_n\ttx \: !\ttx\ttv_n$} (en2);
\path (en3) edge [loop below]		node  {$?0 \: !0$}		     (en3);
\path (en4) edge [loop above,color=red] node  {$?0\ttx \: !\ttx0$}	     (en4);
\path (en5) edge [loop below]		node  {$?1 \: !1$}		     (en5);
\path (en6) edge [loop above,color=red] node  {$?1\ttx \: !\ttx 1$}	     (en6);

\node at (0,2)	  [state]	 (f0)	    {};
\node at (1.5,2)  [state]	 (f11)	    {\tiny $3,1$};
\node at (3.5,2)  [state]	 (f12)	    {};
\node at (5,2)	  [state]	 (f13)	    {};
\node at (6.5,2)  [state]	 (f21)	    {\tiny $3,2$};
\node at (8.5,2)  [state]	 (f22)	    {};
\node at (10,2)	  [state]	 (f23)	    {};
\node at (11.5,2) [state,dashed] (f31ghost) {};
\node at (14.5,2) [state]	 (fn1)	    {\tiny $3,n$};
\node at (16.5,2) [state]	 (fn2)	    {};
\node at (18,2)	  [state]	 (fn3)	    {};

\node at ($(f31ghost)!0.5!(fn1)$) [minimum size=0em,inner sep=0.4em] (Edots) {\Large $\cdots$};

\path (f0) edge [color=red] node  {$?\ttx$} (f11);
\path (f11) edge node {$?\ttv_1 \: !\ttv_1$} (f12);
\path (f12) edge (f13);
\path (f13) edge (f21);
\path (f21) edge node {$?\ttv_2 \: !\ttv_2$} (f22);
\path (f22) edge (f23);
\path (f23) edge [dashed] (f31ghost);
\path (fn1) edge node {$?\ttv_n \: !\ttv_n$} (fn2);
\path (fn2) edge (fn3);

\path (f12) edge [loop above] node  {$?0 \: !0$} (f12);
\path (f13) edge [loop above] node  {$?1 \: !1$} (f13);
\path (f22) edge [loop above] node  {$?0 \: !0$} (f22);
\path (f23) edge [loop above] node  {$?1 \: !1$} (f23);
\path (fn2) edge [loop above] node  {$?0 \: !0$} (fn2);
\path (fn3) edge [loop above] node  {$?1 \: !1$} (fn3);

\node at (-1,10) [minimum size=0em] (lab0)    {L$_0$};
\node at (-1,8)	 [minimum size=0em] (lab1)    {L$_1$};
\node at (-1,6)	 [minimum size=0em] (lab21)   {L$_{2,1}$};
\node at (-1,2)	 [minimum size=0em] (lab3)    {L$_3$};
\node at (-1,0)	 [minimum size=0em] (labrest) {$\vdots$};
\node at (18,0)	 [state]	    (qf)      {$\ttf$};

\node at ($(lab21)!0.5!(lab3)$) [minimum size=0em] (labdots){$\vdots$};
\node [right of=labdots,node distance=2em,label=right:{$\cdots\;\;\;$ Repeat line above $2n-1$ times $\;\;\;\cdots$}]  (moreL3) {};
\node [right of=labrest,node distance=2em,label=right:{$\cdots\;\;\;$ Repeat lines L$_1$ to L$_3$ ($2n+2$ lines each time) for remaining clauses $C_2,\ldots,C_m$ $\;\;\;\cdots$}] (moreL2-4){};

\draw (cnc) .. controls (18,9.5) .. (7,9.5) .. controls (0,9.5) .. (d11);
\draw (dn3) .. controls (18,7.5) .. (7,7.5) .. controls (-1.2,7.5) .. (e0);
\draw (en6) .. controls (18,4.5) .. (7,4.5) .. controls (-1,4.5) and (-1,4) .. (0,4);
\draw [dashed] (10,3.5) .. controls (10,3.5) .. (6,3.5) .. controls (-0.5,3.5) and (-0.5,2.5) .. (f0);
\draw [dashed] (fn3) .. controls (18,1) ..  (9,1) .. controls (-0.5,1) and (-0.5,0.5) .. (-0.1,0.1);
\path (16,0) edge [dashed] (qf);

  \end{tikzpicture}
}
\caption{Single-path LCM for satisfiability of $\phi = (v_2\lor \neg v_3\lor \neg v_n)\land C_2\cdots \land C_m$.}
\label{fig-reduction-singlepath}
\end{figure}
 The operations
are grouped in lines and we describe them informally.
\begin{description}

\item[L$_0$, choosing a valuation nondeterministically:] $S_\phi$
  first write $\ttv_10\ttv_20\ldots\ttv_n0$ on the channel. This is
  our encoding for the valuation that is $0$ for all variables. Then
  $S_\phi$ reads the valuation and write it back, possibly changing any $0$
  value with a $1$ (this happens at the red-coloured actions), and
  thus picking an arbitrary valuation nondeterministically. Here we
  see how the $\ttv_1,\ldots,\ttv_n$ markers are used to check
  positions inside the valuation.

\item[L$_1$, marking where clause $C_1$ is validated:] $S_\phi$ now
  checks whether the valuation stored on the channel makes $C_1$ true. In
  this example, we assume that $C_1$ is $v_2\lor\neg v_3\lor\neg
  v_n$. Again $S_\phi$ reads the valuation and writes it
  back. However, if it reads $\ttv_2 1$ or $\ttv_3 0$ or $\ttv_n 0$, it
  writes it back followed by a special checkmark symbol $\ttx$ that
  ``means $C_1$ has been validated'' (see red actions).  Note that as
  many as 3 occurrences of $\ttx$ can be inserted in the encoding of
  the valuation.

\item[L$_{2,1}$, pushing $\ttx$ to the head of the valuation
encoding:] $S_\phi$ now pushes any checkmark symbol to the left. This
  is done along the $L_{2,1}$ line. While the valuation is read and
  written back as usual (black actions), any symbol preceding a $\ttx$
  can swap position with it (red actions).

\item[L$_{2,2}$, \ldots, L$_{2,2n}$, more pushing $\ttx$ to the left:]
  this behaviour is repeated $2n$ times in total, so that any $\ttx$
  can be pushed completely to the left of the valuation. In case of
  multiple occurrences of $\ttx$, we just need one of them to reach
  the head of the valuation so we assume that the other ones will just
  be lost.

\item[L$_3$, checking that clause $C_1$ has been validated:] Now
  $S_\phi$ knows where to expect $\ttx$. The machine can only proceed
  if indeed a $\ttx$ is present in the channel, in front of the
  valuation, and thus if the valuation on the channel satisfies
  $C_1$. The rest of the line reads and writes back the valuation,
  clearing it of any remaining $\ttx$'s.

\item[Same treatment for the remaining clauses $C_2,\ldots,C_m$:]
$S_\phi$ now continues with similar locations and rules checking that
  the remaining clauses are validated.
\end{description}

Note that, once the valuation has been picked nondeterministically (in
L$_1$), it cannot be modified. Also note that the machine will block
if one of the $\ttv_i$ markers is lost before the last clause has been
validated. If one of the $0/1$ values of the valuation is lost, this
value cannot be used any more for checkmarking a validated
clause. Such message losses do not lead to any incorrect behaviour,
they can only hinder the validation of a clause.

Finally, starting from $(0,\epsilon)$, $S_\phi$ can reach its final
location $\ttf$ iff $\phi$ is satisfiable.
\\

Now the reduction extends to show prove $\NP$-hardness of
unboundedness for single-path LCMs with exactly the same adaptation as
in the proof for acyclic LCMs. For hardness of nontermination
a little more work is needed since every cycle where $S_\phi$ reads the
valuation and writes it back could become a nonterminating cycle if all
but one letter are lost. One possible trick to overcome this is to have two
copies of the alphabet, say of two different colours, and to ensure
that in all its phases the machine reads in one colour and writes back
in the other, so that the valuation is always read and written in
alternating colours. Once this is implemented, the system cannot have
infinite runs as is. Adding a single loop on $\ttf$, the
final control location, as we did
for acyclic LCMs, now provides a correct reduction
from $\SAT$ to nontermination for single-path LCMs.
\\

The idea behind this reduction can easily be adapted so that it
applies to single-path VASSes and lossy VASSes, or equivalently, to
channel machines with a  unary alphabet. One uses
$2n$ channels (or counters) for storing the valuation and $m$ distinct
counters for marking the clauses that have been validated.

Restricting to a binary alphabet on a single channel is equally easy
for reliable FIFO automata, but more difficult when message losses have to
be taken care of. Therefore we won't attempt it in this preliminary version.

\section{Multiple channels}
\label{app-multiple-channels}

The analysis we conducted in
\Cref{sec-backward-reachability} carries over
without any difficulty to systems with multiple channels.
\Cref{lem-up-sigma-um} and \Cref{thm-yk-pk-lk} remain valid since,
once $\sigma$ and $k$ have been fixed, computing
$\pr[\sigma^k](\tuple{x_1,\ldots,x_c})$ for a system with $c$ channels
can be done independently for each of the $c$ channels: one only needs
to distribute the actions on $\sigma$ to their corresponding channel,
so that $\rea(\sigma)$ now is some tuple $\tuple{u_1,\ldots,u_c}$. In
particular the bound in \Cref{thm-L-sigma-x} becomes
\[
L\bigl(\sigma,\tuple{x_1,\ldots,x_c}\bigr) \leq
\max\nolimits_{i=1}^c |x_i| \cdot (|u_i|+1)\:,
\;\;\text{ where }
 \rea(\sigma)=\tuple{u_1,\ldots,u_c}
\:.
\]

\end{document}